\documentclass[11pt]{article}
\usepackage{amsthm,amsmath,amssymb,amsfonts}
\usepackage{epsfig}
\usepackage{latexsym,nicefrac,bbm}
\usepackage{xspace}
\usepackage{color,fancybox,graphicx,subfigure,fullpage}
\usepackage[top=1in, bottom=1in, left=1in, right=1in]{geometry}
\usepackage{tabularx}
\usepackage{hyperref}
\usepackage{pdfsync}
\usepackage{multicol}
\usepackage{cite,cleveref}

\renewcommand{\epsilon}{\varepsilon}

\newcommand{\nfrac}{\nicefrac}

\newcommand{\cM}{\mathcal{M}}

\newcommand{\cU}{\mathcal{U}}

\newcommand{\st}{\mathrm{s.t.}}

\newcommand{\cR}{\mathbb{R}}

\usepackage{bbm}
\newcommand{\white}[1]{ \textcolor{white} {#1}}
\usepackage{enumitem}
\newcommand\labelthis[1]{\addtocounter{equation}{1}\tag{{#1},\ \theequation}}
\usepackage{subfiles}
\usepackage{mathtools}
\newcommand\numberthis{\addtocounter{equation}{1}\tag{\theequation}}
\DeclareMathOperator*{\eE}{\mathbb{E}}
\newcommand{\argmax}{\operatornamewithlimits{argmax}}
\newcommand{\cP}{\mathcal{P}}
\newcommand{\cL}{\mathcal{L}}
\newcommand{\cQ}{\mathcal{Q}}

\usepackage{epsfig}
\usepackage{verbatim}

\theoremstyle{definition}

\usepackage{algorithmic}
\usepackage{algorithm,caption}

\usepackage{mhequ}
\def \be{\begin{equs}}
\def \ee{\end{equs}}

\newtheorem{theorem}{Theorem}[section]
\newtheorem{lemma}[theorem]{Lemma}

\newtheorem*{theorem*}{Theorem}
\newtheorem{remark}[theorem]{Remark}

\crefname{theorem}{Theorem}{Theorems}
\crefname{observation}{Observation}{Observations}
\crefname{proposition}{Proposition}{Propositions}
\crefname{claim}{Claim}{Claims}
\crefname{condition}{Condition}{Conditions}
\crefname{example}{Example}{Examples}
\crefname{fact}{Fact}{Facts}
\crefname{lemma}{Lemma}{Lemmas}
\crefname{corollary}{Corollary}{Corollaries}
\crefname{definition}{Definition}{Definitions}
\crefname{remark}{Remark}{Remarks}

\title{Toward Controlling Discrimination in Online Ad Auctions\footnote{The code for the simulations is available at \url{https://github.com/AnayMehrotra/Fair-Online-Advertising}.}}

\usepackage{authblk}

\author[1]{L. Elisa Celis}
\author[2]{Anay Mehrotra}
\author[3]{Nisheeth K. Vishnoi}
\affil[1,3]{\small Yale University}
\affil[2]{\small Indian Institute of Technology, Kanpur}
\date{}

\begin{document}

\maketitle

\thispagestyle{empty}
\begin{abstract}
  Online advertising platforms are thriving due to the customizable audiences they offer advertisers.
  However, recent studies show that advertisements can be discriminatory with respect to the gender or race of the audience that sees the ad, and may inadvertently cross ethical and/or legal boundaries.
  To prevent this, we propose a constrained ad auction framework that maximizes the platform’s revenue conditioned on ensuring that the audience seeing an advertiser’s ad is distributed appropriately across sensitive types such as gender or race.
  Building upon Myerson’s classic work, we first present an optimal auction mechanism for a large class of fairness constraints.
  Finding the parameters of this optimal auction, however, turns out to be a non-convex problem.
  We show that this non-convex problem can be reformulated as a more structured non-convex problem with no saddle points or local-maxima; this allows us to develop a gradient-descent-based algorithm to solve it.
  Our empirical results on the A1 Yahoo! dataset demonstrate that our algorithm can obtain uniform coverage across different user types for each advertiser at a minor loss to the revenue of the platform, and a small change to the size of the audience each advertiser reaches.
\end{abstract}

\newpage

\thispagestyle{empty}
\tableofcontents
\newpage

\section{Introduction}
  Online advertisements are the main source of revenue for social-networking sites and search engines such as Google~\cite{google_revenue}.
  Ad exchange platforms allow advertisers to select the target audience for their ad by specifying desired user demographics, interests and browsing histories~\cite{ad_targeting}.
  Every time a user loads a webpage or enters a search term, bids are collected from relevant advertisers~\cite{relevant_advertisers}, and an auction is conducted to determine which ad is shown, and how much the advertiser is charged~\cite{Muthukrishnan09,yuan_et_al,varian2007position}.
  As it is not practical for advertisers to place individual bids for every user, the advertiser instead gives some high-level preferences about their budget and target audience, and the platform places bids on their behalf~\cite{auto_bid}.

  More formally, let there be $n$ advertisers, and $m$ types of users.
  Each advertiser $i$ specifies their target demographic, average bid, and budget to the platform, which then decides a distribution, $\cP_{ij}$, of bids of advertiser $i \in [n]$ for user type $j \in [m]$.
  These distributions represent the value of the user to the advertiser, and ensure that the advertiser only bids for users in their target demographic, with the expected bid not exceeding the amount specified by the advertiser~\cite{bids_represent_value}.
  At each time step, a user visits a web page (e.g., Facebook or Twitter), the user's type $j\in [m]$ is observed, and a bid $v_i$ is drawn from $\cP_{ij}$, for each advertiser $i\in[n]$.
  Receiving these bids as input, the mechanism $\mathcal{M}$ decides an allocation $x(v)$ and price $p(v)$ for the advertisement slot.
  Several Ad Exchanges including Google Ads~\cite{snd_price_google2} and Facebook Ads~\cite{snd_price_in_practice},
  use variants of second price auction mechanism~\cite{ostrovsky2011reserve}\footnote{If the auction sells a single item, then Myerson's mechanism~\cite{myer} reduces to a second price auction mechanism with a reserve price~\cite{treatise_on_mechanisms}.}.

  Overall, such targeted advertising leads to higher utilities for the advertisers who show content to relevant audiences, for the users who view related advertisements, and for the platform which can benefit from selling targeted advertisements~\cite{farahat,yan_et_al,creepy_or_cool,goldfarb_tucker}.
  However, targeted advertising can also lead to discriminatory practices.
  For instance,  searches with ``black-sounding" names were much more likely to be shown ads suggestive of an arrest record~\cite{sweeney13}.
  Another study found that women were shown fewer advertisements for high paying jobs than men with similar profiles~\cite{datta}.
  In fact, recent experiments demonstrate that ads can be \emph{inadvertently} discriminatory;
  \cite{lambrecht_tucker} found that STEM job ads, specifically designed to be unbiased by the advertisers, were shown to more men than women across all major platforms (Facebook Ads, Google Ads, Instagram and Twitter).
  On Facebook, a platform with $52\%$ women~\cite{audience_insights} the advertisement was shown to $20\%$ more men than women.
  \cite{ali2019discrimination} find that this could be a result of \textit{competitive spillovers} among advertisers, and is neither a pure reflection of pre-existing cultural bias, nor a result of user input to the algorithm.
  Such (likely inadvertent) discrimination has led to two recent cases filed against Facebook, which will potentially lead to civil lawsuits alleging employment and housing discrimination~\cite{facebook_gender_bias,facebook_housing_1,facebook_housing_2,facebook_job_ethinicity}.

  To gain intuition on how inadvertent discrimination could happen, consider the setting in which there are two advertisers with similar bids/budgets, but one advertiser specifically targets women (which is allowed for certain types of ads, e.g., related to clothing), while the second advertiser  does not target based on gender (e.g., because they are advertising a job).
  The first advertiser creates an imbalance on the platform by taking up ad slots for women and, as a consequence, the second advertiser ends up advertising to disproportionately fewer women and is  inadvertently discriminatory.
  Currently, online advertising platforms have no mechanism to check this type of discrimination.
  In fact, the only way around this would be for the advertiser to set up separate campaigns for different user types and ensure that each campaign reached a similar number of the sub-target audience.
  However, online platforms often reject such campaigns in the apprehension of discriminatory practices~\cite{lambrecht_tucker, discriminatory_practice}.

  \begin{figure}[t]
    \footnotesize
    \centering
    \hspace{-4mm}
    \subfigure[
    \footnotesize
    Only the first advertiser bids for this user.
    ]
    {\includegraphics[clip, trim=2.0cm 22.9cm 8.4cm 3.1cm,width=8.2cm]{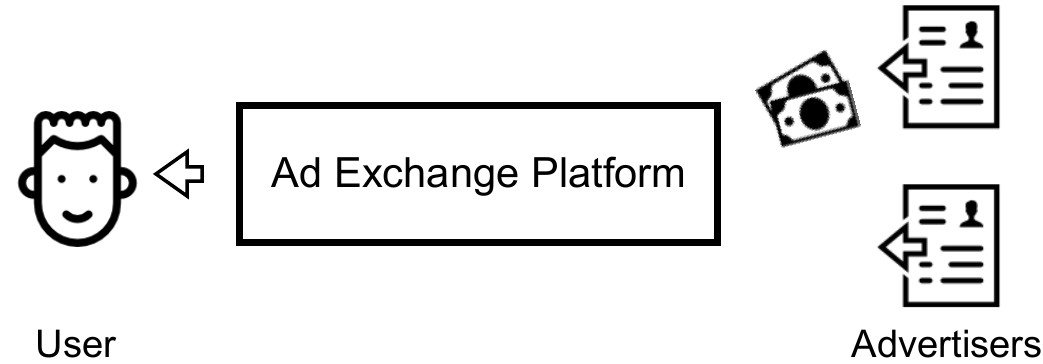}}
    \hspace{2mm}\subfigure[
    \footnotesize
    Both advertisers place bids for this user.
    ]
    {\includegraphics[clip, trim=2.0cm 18.25cm 8.4cm 7.75cm, width=8.2cm]{figures/entities_on_ad_platforms.pdf}}
    \caption{
    When a user visits the platform, the platform accepts bids from {\em relevant} advertisers, and conducts an auction to decide the ad to be shown to the user.
    Different advertisers have different target audiences and only bid for users in their target audience.
    }
    \label{fig:entities_on_ad_platforms}
  \end{figure}

  \subsection*{Our Contributions}
    Our main contribution is an optimization-based framework which maximizes the revenue of the platform subject to satisfying constraints that prevent the emergence of inadvertent discrimination as described above.
    The constraints can be formulated as any one of a wide class of ``group fairness'' constraints as presented in~\cite{CKSV18}, which constrains the distribution of an ad’s audience across the sensitive types to ensure proportionality across types as defined by the platform.
    The framework allows for intersectionality, allowing constraints across multiple sensitive attributes (e.g., gender, race, geography and economic class) and allows for restricting different advertisers to different constraints.

    Formally, building on Myerson's seminal work~\cite{myer}, we characterize the truthful revenue-optimal mechanism which satisfies the given  constraints (Theorem~\ref{thm:1}).
    The user types, as defined by their sensitive attributes, are taken as input along with the type-specific bid distributions for each advertiser, and we assume that bids are drawn from these distributions independently.
    Our mechanism is parameterized by constant ``shifts" which it applies to bids for each advertiser-type pair.
    Finding the parameters of this optimal mechanism, however, is a non-convex optimization problem, both in the objective and the constraints.
    Towards solving this,
    we first propose a novel reformulation of the objective as a composition of a convex function constrained on a polytope, and an unconstrained non-convex function (Theorem~\ref{thm:2New}).
    Interestingly, the non-convex function is reasonably well behaved, with no saddle-points or local-maxima.
    This allows us to develop a gradient descent based scheme (Algorithm~\ref{Algorithm1}) to solve the reformulated program, which under mild \hyperref[assumptions]{assumptions} has a fast  convergence rate of $\widetilde{O}(\nfrac{1}{\epsilon^2})$ (Theorem~\ref{thm:2}).

    We evaluate our approach empirically by studying the effect of the constraints on the revenue of the platform and the advertisers using the \textit{Yahoo! Search Marketing Advertising Bidding Data}~\cite{yah}. We find that our mechanism can obtain uniform coverage across different user types for each advertiser while losing less than 5\% of the revenue (Figure~\ref{fig:Exp1Result1}).
    Further, we observe that the total-variation distance between the fair and unconstrained distributions of total advertisements an advertiser shows on the platform is less than $0.05$ (Figure~\ref{fig:Exp1Result2}).

    To the best of our knowledge, we are the first to give a framework to prevent inadvertent discrimination in online ad auctions.

    \section{Our Model}\label{sec:ourModel}

    \subsection{Preliminaries}
      In this section, we provide some key preliminaries.
      For a detailed discussion, we refer the reader to the excellent treatises \cite{treatise_on_mechanisms,nisan_roughgarden_tardos_vazirani_2007} on Mechanism design.
      A {mechanism} $\mathcal{M}$ is defined by its allocation rule $x\colon \cR^n \to [0,1]^n$, and its payment rule $p\colon \cR^n \to \cR_{\geq 0}^{n}$.
      {\em Truthful mechanisms} are those in which revealing the true valuation is optimal for all bidders.
      Further, the can be shown that the allocation rule $x(b_1,b_2,\dots,b_n)$, of any truthful mechanism must be monotone in $b_i$ for all $i\in[n]$.
      \cite{revelation} proved for any mechanism $\cM$ there exists a truthful mechanism $\tau(\cM)$ such that $\tau(\cM)$ offers the same revenue to the seller and the same utility to each bidder as $\cM$.
      As such, we restrict ourselves to truthful mechanisms.
      Furthermore, it is a well known fact \cite{nisan_roughgarden_tardos_vazirani_2007} that for any truthful mechanism its payment rule $p$, is uniquely defined by its allocation rule $x$.
      Hence, for any truthful mechanism our only concern is the allocation rule $x$.

      Let $\cP$ be the distribution of valuation of a bidder, ${\rm pdf}\colon \cR\to\cR_{>0}$ be its probability density function, and ${\rm cdf}\colon \cR\to[0,1]$ be its cumulative density function, then we define the {\em virtual valuation} $\phi\colon\mathrm{supp}(\cP)\to \cR$, as $\phi(v) \coloneqq v - (1-{\rm cdf}(v))({\rm pdf}(v))^{-1}$.
      We say $\cP$ is {\em regular} if $\phi(v)$ is non-decreasing in $v$.
      Likewise, we say $\cP$ is {\em strictly regular} if $\phi(v)$ is strictly increasing in $v$.
      \paragraph{Myerson's Optimal Mechanism.}
        Myerson's mechanism is defined as the VCG mechanism~\cite{clarke,groves,vickrey} where the virtual valuation $\phi_i$, is submitted as the bid $v_i$ for each bidder $i$.
        If the valuations $v_i$, and therefore, the virtual valuations $\phi_i$ are \textit{independent}, then for any truthful mechanism the virtual surplus $\sum_{i\in [n]}\phi_i x_i(\phi_i)$, is equal to the revenue in expectation over the bids.
        Since VCG is surplus maximizing, if Myerson's mechanism is truthful then it maximizes the revenue.
        It can be shown that if the bids have a regular distribution, then Myerson's mechanism is truthful, and therefore, revenue maximizing.

      \paragraph{Notation.}
        Let $\phi_{ij}\in \cR$ be the virtual valuation of advertiser $i \in [n]$ for type $j \in [m]$, $f_{ij}\colon \cR\to \cR_{\geq 0}$ be its probability density function, and $F_{ij}\colon \cR\to [0,1]$ be its cumulative density function.
        We denote the joint virtual valuation of all advertisers for type $j$ by $\phi_j\in \cR^{n}$, and its joint probability density function by $f_{j}\colon \cR^{n}\to \cR_{\geq 0}$.
        The types $j\in[m]$ are distributed according to a known distribution $\cU$.
        Finally, given a user of type $j$, let a mechanism's allocation rule be $x_j\colon\cR^n\to[0,1]^n$.

      \subsection{Fairness Constraints}
        We would like to guarantee that advertisers have a fair coverage across user types.
        We do so by placing constraints on the coverage of an advertiser.
        Formally, we define advertiser $i$'s coverage of type $j$, $q_{ij}$, as the joint probability that advertiser $i$ wins the auction and the user is of type $j$
        \begin{align*}
          \label{def:coverage1}
          &\hspace{-4mm}q_{ij}(x_j)\coloneqq \Pr\nolimits_\cU[j]\hspace{-3mm}\int\limits_{{\rm supp}(\phi_j)} \hspace{-4mm}x_{ij}(\phi_j)df_j(\phi_j), \labelthis{Coverage}
        \end{align*}
        where $x_{ij}(\phi_j)$ is the $i$-th component of $x_{j}(\phi_j)$.
        Then, we consider the proportional coverage of the advertiser on each type.
        Given vectors $\ell_j$, $u_j$ $\in [0,1]^n \ \forall \ j \in [m]$, we define $(\ell,u)$-fairness constraints for each advertiser $i$ and type $j$, as a lower bound $\ell_{ij}$, and an upper bound $u_{ij}$, on the proportion of users of type $j$ the advertiser shows ads to, i.e., we impose the following constraints for all $i\in [n]$ and $j\in[m]$
        \begin{align}\label{eq:constraints}
          &\hspace{51.5mm}\ell_{ij} \leq \frac{q_{ij}  }{\sum\nolimits_{t=1}^{m}q_{it}  } \leq u_{ij}.
          \labelthis{$(\ell,u)$-fairness constraints}
        \end{align}
    \subsection{Discussion of Fairness Constraints}
      Returning to the example presented in the introduction, we can ensure that the advertiser shows  $x$\% of total ads to women, by choosing a lower bound of $x$ for this advertiser on women.
      More generally, for $m$ user types, moderately placed lower bounds and upper bounds ($\ell_{ij} \sim \nfrac{1}{m}$ and $u_{ij} \sim \nfrac{1}{m}$), for some subset of advertisers, ensure this subset has a  uniform coverage across all types, while allowing other advertisers to target specific types.

      Importantly, while ensuring fairness across multiple types our constraints allow for targeting within any single type.
      This is vital as the advertiser may not derive the same utility from each user, and could be willing to pay a higher amount for more relevant users in the same type.
      For example, if the advertiser is displaying job ads, then a user already looking for job opportunities may be of a higher value to the advertiser than one who is not.

      For a detailed discussion on how such constraints can encapsulate other popular metrics, such as statistical parity, we refer the reader to~\cite{CelisHKV19}.
    \subsection{Optimization Problem}
      We would like to develop a mechanism which maximizes the revenue while satisfying the upper and lower bound constraints in
      Eq.~(\hyperref[eq:constraints]{2}).
      Towards formally stating our problem, we define the revenue of mechanism $\mathcal{M}$, with an allocation rule $x_j\colon \cR^{n} \to[0,1]^{n}$ for type $j$ as
      \begin{align}
        \label{def:revenue1}
        &\hspace{-2mm}\mathrm{rev}_{\mathcal{M}}\coloneqq\sum_{\substack{i\in [n],\ j\in [m]}}\Pr\nolimits_{\cU}[j]\hspace{-3mm}\int\limits_{{\rm supp}(\phi_j)}\hspace{-5mm}\phi_{ij}x_{ij}(\phi_j)df_j(\phi_j), \labelthis{Revenue}
      \end{align}
      where $x_{ij}(\phi_j)$ and $\phi_{ij}$ are the $i$-th component of $x_{j}(\phi_j)$ and $\phi_{j}$ respectively.
      Thus, we can express our optimization problem with respect to {\em functions} $x(\cdot)$, or as an infinite dimensional optimization problem as follows.
      \noindent{(\bf Infinite-dimensional fair advertising problem). }For all user types $j\in[m]$, find the optimal allocation rule $x_{j}(\cdot)\colon \cR^{n} \to [0,1]^n$ for
      \begin{align}
        \label{problem:primal}
        &\hspace{-8mm}\hspace{25mm}\hspace{-1mm}\max_{x_{ij}(\cdot)\geq 0} \mathrm{rev}_{\mathcal{M}}(x_1,x_2,\dots,x_m) \quad\quad \numberthis& \\
        &\hspace{-8mm}\hspace{28.0mm}\mathrm{s.t.}   \label{eq:primalLowerbound}  \hspace{3mm} q_{ij}(x_j)\geq \ell_{ij}\sum_{t=1}^{m} q_{it}(x_t)  \ \hspace{0.5mm}\forall \ j \in [m], i \in [n]   \\
        &\hspace{-8mm} \hspace{28.0mm}\hspace{8mm} q_{ij}(x_j)\leq u_{ij}\sum_{t=1}^{m}q_{it}(x_t)   \
        \forall \ j \in [m], i \in[n] \label{eq:primalupper bound} \\
        &\hspace{-8mm}\hspace{28.0mm} \hspace{7mm}\sum\limits_{i=1}^{n}x_{ij}(\phi_j) \leq\ 1  \hspace{14mm} \forall \ j \in [m], \phi_j,  \label{eq:primalSingleItem}
      \end{align}
      where \eqref{eq:primalLowerbound} encodes the lower bound constraints, \eqref{eq:primalupper bound} encodes the upper bound constraints, and \eqref{eq:primalSingleItem} ensures that only one ad is allocated.
      In the above problem, we are looking for a collection of optimal continuous functions $x^\star$.
      To be able to solve this problem, we need -- in the least -- a finite dimensional formulation of the fair online advertisement problem.

      \section{Other Related Work}
        \cite{DworkI19} consider a framework which selects an ad category (e.g., job or housing) every time a user visits the platform.
        Given fair mechanisms for each category, they construct a {\em fair composition} of these mechanisms.
        However, they do not show how to design fair mechanisms for each category, or study how the composition affects the platform's ad revenue.
        Another related problem is to design optimal mechanisms which satisfy contract constraints~\cite{ghosh09,BalseiroFMM14,Pai2012Auction};
        these constraints allocate a minimum number of ad spots to advertisers with a contract,
        and are different from our constraints which control the fraction of each sensitive type the ads are shown to.

        Several prior works address the problems of polarization and algorithmic bias,
        including~\cite{GarimellaMGM18,CKSKV19}
        who control polarization in social-networks and personalized feeds,
        \cite{PanigrahiSAT12} who diversify personal feeds,
        and \cite{dpp_data_summarization} who create a diverse and balanced summary of a set of results.
        In addition, ~\cite{RadlinskiKJ08,asudeh_stoyanovich,CSV18} study fair ranking algorithms; these could be used to generate a balanced list of results on job platforms and other search engines.
        While these works are related to our broad goal of controlling algorithmic bias, their formulation is different since they do not involve a bidding mechanism.
        Therefore, their solutions cannot be applied to our problem.

        Finally, a framework approach to fairness constraints has shown to be effective in various other applications such as classification~\cite{CelisHKV19, HV19, ZafarVGG17}, selection of representatives~\cite{CelisHV18}, and personalization~\cite{CV17}.

        \section{Theoretical Results}\label{sec:main_results}
          Our first result is structural, and gives a characterization of the optimal solution $x^\star$, to the {infinite-dimensional fair advertising problem}, in terms of a matrix $\alpha\in\cR^{n\times m}$, making it a finite-dimensional optimization problem with respect to $\alpha$.
          \begin{theorem}
            \label{thm:1}
            \textbf{(Characterization of an optimal allocation rule).}
            There exists an $\alpha = \{\alpha_j\}_{j\in[m]} \in \cR^{n\times m}$ such that
            if for all $j\in [m]$, $\cP_{j}$ are strictly regular and independent,
            then the set of  allocation rules $x_j(\cdot,\alpha_j)\colon \cR^{n}\to[0,1]^{n}\ \forall \ j\in[m]$, defined below,
            is optimal for the infinite-dimensional fair advertising problem
            \begin{align*}
              \hspace{35mm}
              x_{ij}(v_j,\alpha_j)\  \coloneqq \   \mathbb{I}[\ i\in\argmax\nolimits_{\ell \in[n]}(\phi_{\ell j}(v_{\ell j})+\alpha_{\ell j})\ ].
              \labelthis{$\alpha$-shifted mechanism}\label{def:alpha_shifted_mechanism}
            \end{align*}
            Where we randomly breaks ties if any (this is equivalent to the allocation rule of the VCG mechanism).
          \end{theorem}

          We present the proof of Theorem~\ref{thm:1} in Section~\ref{sec:proof1}.
          In the proof, we analyze the dual of the {infinite-dimensional fair advertising problem}.
          We reduce the dual problem to one lagrangian variable, by fixing the lagrangian variables corresponding lower bound~(\hyperref[eq:primalLowerbound]{5}) and upper bound~(\hyperref[eq:primalupper bound]{6}) constraints to their optimal values.
          The resulting problem turns out to be the dual of the unconstrained revenue maximizing problem, for which Myerson's mechanism is the optimal solution.
          We interpret the fixed lagrangian variables as shifting the original virtual valuations $\phi_{ij}$.
          It then follows that for some shift $\alpha\in\cR^{n\times m}$, the  $\alpha$-shifted mechanism~(\hyperref[def:alpha_shifted_mechanism]{8}) is the optimal solution to the { infinite-dimensional fair advertising problem}.

          Now, our task is reduced from finding an optimal allocation rule, to finding an $\alpha$ characterizing the optimal allocation rule.
          Towards this, let us define the revenue, $\mathrm{rev}_{\mathrm{shift}}\colon \cR^{n\times m}\to \cR$ and coverage $q_{ij}\colon \cR^{n\times m}\to [0,1]$ as functions of $\alpha$
          \begin{align*}
            &\hspace{-2mm}\mathrm{rev}_{\mathrm{shift}}(\alpha)\coloneqq\sum\limits_{\substack{i\in[n]\\
            j\in[m]}}\Pr_{\cU}[j]\hspace{-3mm}\int\limits_{\mathrm{supp}(f_{ij})}
            \hspace{-4mm}y f_{ij}(y)\hspace{-1.0mm}\prod\limits_{k\in [n]\backslash \{i\}}\hspace{-2.5mm}F_{kj}(y+
            \alpha_{ij}-\alpha_{kj})dy\labelthis{Revenue $\alpha$-shifted mechanism}\label{eq:revenue_as_function_of_alpha}\\
            &\hspace{4.0mm}q_{ij}(\alpha)\coloneqq\hspace{9.5mm}\Pr_{\cU}[j]\hspace{-3mm}\int\limits_{\mathrm{supp}(f_{ij})}
            \hspace{-4mm} f_{ij}(y)\hspace{-1.0mm}\prod\limits_{k\in [n]\backslash \{i\}}\hspace{-2.5mm}F_{kj}(y+\alpha_{ij}-\alpha_{kj})dy.
            \labelthis{Coverage $\alpha$-shifted mechanism}
            \label{def:coverage2}
          \end{align*}
          These follow by observing that (\hyperref[def:alpha_shifted_mechanism]{8}) selects the advertiser with the highest shifted virtual valuation, and then using this allocation rule in Eq.~(\hyperref[def:revenue1]{3}) and Eq.~(\hyperref[def:coverage1]{1}) respectively.
          Depending on the nature of the distribution, the gradients $\nfrac{\partial \mathrm{rev}_{\mathrm{shift}}(\alpha)}{\partial\alpha_i}$ and $\nfrac{\partial q_{ij}(\alpha)}{\partial\alpha_i}$ may not be monotone in $\alpha$ (e.g., consider the exponential distribution).
          Therefore, in general neither is ${\rm rev}_{\rm shift}(\cdot)$ a concave, nor is $q_{ij}(\cdot)$ a convex function of $\alpha$ (see Section~\ref{app:revenue_is_non_concave} for a concrete example).
          Hence, this optimization problem is non-convex both in its objective and in its constraints.
          We require further insights to solve the problem efficiently.

          Towards this, we observe that revenue is a concave function of $q$.
          Consider two optimal allocation rules obtaining coverages $q_1,\ q_2\in [0,1]^{n\times m}$ and revenues $R_1,\ R_2\in \cR$ respectively.
          If we use the first with probability $\gamma \in [0,1]$, we achieve a coverage $\gamma q_1+(1-\gamma)q_2$ with revenue $\gamma R_1 + (1-\gamma)R_2$.
          Therefore, the optimal allocation rule achieving $\gamma q_1+(1-\gamma)q_2$ has a revenue of at least $\gamma R_1 + (1-\gamma)R_2$.
          This shows that for optimal allocation rules revenue is a concave function of the coverage $q$.

          Let $\mathrm{rev}\colon  [0,1]^{(n-1)\times m}\to \cR$, be the maximum revenue of the platform as a function of coverage $q$.\footnote{We drop $q_{ij}$ for one $i\in[n]$ and each $j\in[m]$. This is crucial to calculate $\nabla {\rm rev}$ (see Remark~\ref{rem:fix_shift_for_inverting_jacobian}).
          By some abuse of notation we write ${\rm rev}(q)$ for $q\in \cR^{n\times m}$ instead of using $q\in \cR^{(n-1)\cdot m}$.}
          Consider the following two optimization problems.

          \vspace{2mm}

          \noindent{\bf(Optimal coverage problem).}
          Find the optimal $q \in[0,1]^{n\times m}$ for,
          \begin{align*}
            \label{problem:convex_problem}
            &\hspace{-5mm}\hspace{36mm}\max_{q\in[0,1]^{n}}\  \mathrm{rev}(q) \numberthis& \\
            &\hspace{-5mm}\hspace{40.5mm}\mathrm{s.t.}
            \label{eq:primalLowerbound2}
            \hspace{1mm}\hspace{2.5mm} q_{ij} \geq \ell_{ij}\sum_{t=1}^{m}q_{it}  \hspace{11mm} \forall \ j \in [m], i \in [n]\numberthis\\
            \label{eq:primalupper bound2}
            &\hspace{-5mm}\hspace{40.5mm}\hspace{6.0mm}\hspace{2.5mm} q_{ij}\leq u_{ij}\sum_{t=1}^{m}q_{it}  \hspace{10.6mm}
            \forall \ j \in[m], i \in [n] \numberthis\\
            \label{eq:primalSingleItem2}
            &\hspace{-5mm} \hspace{47.5mm}\hspace{-1.5mm}\hspace{2.5mm}\sum\limits_{i=1}^{n}q_{ij} \leq\ \Pr\nolimits_{\cU}[j]  \ \hspace{7.0mm} \forall \ j \in [m].\numberthis
          \end{align*}
          \noindent{\bf(Optimal shift problem).}
          Given the target coverage $\delta\hspace{-0.50mm}\in\hspace{-0.50mm}[0,1]^{n\times m}$, find the optimal $\alpha \hspace{-0.50mm}\in\hspace{-0.50mm}\cR^{n\times m}$ for
          \begin{align*}
            \label{problem:nonconvex_unconstrained}
            \min_{\alpha\in \cR^{n\times m}}\cL(\alpha)\coloneqq\|\delta- q(\alpha)\|_F^2.\numberthis
          \end{align*}

          \noindent Our next result relates the solution of the above two problems with the {infinite-dimensional fair advertising problem}.

          \begin{theorem}
            \label{thm:2New}
            Given a solution $q^\star\in[0,1]^{n\times m}$ to the optimal coverage problem, the solution $\alpha^\star$ to the optimal shift problem with $\delta=q^\star$, defines an optimal $\alpha$-shifted mechanism~(\hyperref[def:alpha_shifted_mechanism]{8}) for the infinite-dimensional fair advertising problem.
          \end{theorem}
          \begin{proof}
            For any $j\in [m]$ adding the all $1$ vector, $1_{n}$, to $\alpha_j$ does not change the allocation rule
            in (\hyperref[def:alpha_shifted_mechanism]{8}).
            Thus, it suffices to show that for all $\delta\in [0,1]^{n\times m}$, there is a unique $\alpha$ with $\alpha_{1j}=0\ \forall \ j\in[m]$, such that $q(\alpha)=\delta$.

            We can show that for all $\delta\in[0,1]^{n\times m}$, there is at-least one $\alpha\in\cR^{n\times m}$ such that $q(\alpha)=\delta$.
            In fact, the greedy algorithm which increases all $\alpha_{ij}$, where $q_{ij}(\alpha)<\delta_{ij}$ and $i\neq1$, will find the required $\alpha$.

            To prove it is unique consider distinct $\alpha,\beta\in\cR^{n\times m}$ such that $\alpha_{1j}=\beta_{1j}=0$.
            We can show that $q(\alpha)\neq q(\beta)$.
            In particular, that $q_{i^\prime j^\prime}(\alpha)\neq q_{i^\prime j^\prime}(\beta)$ for $(i^\prime,j^\prime)= \argmax_{i\in[n],\ j\in[m]}|\alpha_{ij}-\beta_{ij}|$.
            Now, the uniqueness of $\alpha$ follows by contradiction.
          \end{proof}

          The above theorem allows us to find the optimal $\alpha$ by solving the {\em optimal coverage} and {\em optimal shift} problems.
          First, let us consider the {\em optimal coverage problem}.
          We already know that its objective is concave.
          We can further observe that its constraints are linear in $q$, and
          in particular, they define a constraint-polytope $\mathcal{Q}\subseteq [0,1]^{n\times m}$.
          Therefore, it is a convex program, and one approach to solve it is to use gradient-based algorithms.

          The problem is that we do not have access to $\nabla {\rm rev}$.
          The key idea is that if we let $\alpha=q^{-1}(\delta)$, then we can calculate $\nabla {\rm rev}(\delta)$ by solving the following linear-system,
          \begin{align*}
            &\hspace{36mm}(J_q(\alpha))^\top \nabla \mathrm{rev}(\delta) = \nabla \mathrm{rev}_{\mathrm{shift}}(\alpha),\labelthis{Gradient Oracle}\label{eq:linearSystem}
          \end{align*}
          where $J_q(\alpha)$ is the Jacobian of $\mathrm{vec}(q(\alpha))\in \cR^{(n-1)m}$ \footnote{${\rm vec(\cdot)}$ represents the vectorization operator.}, with respect to $\mathrm{vec}(\alpha)\in \cR^{(n-1)m}$.
          It turns out that $J_q(\alpha)$ is invertible for all $\alpha\in\cR^{n\times m}$(see Section~\ref{sec:proof_overview_calculating_the_gradient}), and therefore, the above linear-system has an exact solution.

          Now, let us consider the {\em optimal shift problem}.
          Its objective is non-convex (see Figure~\ref{fig:8b}).
          $\nabla \cL(\alpha)$ is a linear combination of $\nabla q_{ij}(\alpha)$ for all $i\in [n]$ and $j\in[m]$.
          Since $J_q(\alpha)$ is invertible, its rows $\{\nabla q_{ij}(\alpha)\}$, are linearly independent, and the gradient is never zero unless we are at the global minimum where $\alpha=q^{-1}(\delta)$.
          This guarantees that the objective does not have a saddle-point or local-maximum, and that any local-minimum is a global minimum.
          Using this we can develop an efficient algorithm to solve the {\em optimal coverage problem} (Lemma~\ref{thm:3}).

          This brings us to our main algorithmic result, which is an algorithm to find the optimal allocation rule for the {\em infinite-dimensional fair advertising problem}.
          \begin{theorem}\label{thm:2}
            \textbf{(An algorithm to solve the infinite-dimensional fair advertising problem).}
            There is an algorithm (Algorithm \ref{Algorithm1}) which outputs $\alpha\in \cR^{n\times m}$ $\text{such that
            if}$ assumptions~(\hyperref[asmp:eta_coverage]{17}), (\hyperref[asmp:distributed_distribution]{18}), (\hyperref[asmp:lipschitz_distribution]{19}),  and (\hyperref[asmp:bounded_bid]{20}) are satisfied,
            the $\alpha$-shifted mechanism~(\hyperref[def:alpha_shifted_mechanism]{8}) achieves a revenue $\epsilon$-close to the optimal for the infinite-dimensional fair advertising problem in
            \begin{align*}
              \widetilde{O}\bigg(\frac{n^7\log m}{\epsilon^2} \frac{(\mu_{\max}\rho)^2}{(\mu_{\min}\eta)^4}
              (L+n^2\mu_{\max}^2)\bigg) \ \mathrm{steps.}
            \end{align*}
            Where the arithmetic calculations in each step are bounded by calculating
            $\nabla \mathrm{rev}$ once
            and $\widetilde{O}$ hides $\log$ factors in $n,\ \rho,\ \eta,\ \mu_{\max}, \ \nfrac{1}{\epsilon}$ and $\nfrac{1}{\mu_{\min}}$.
          \end{theorem}
          \noindent Roughly, the above algorithm has a convergence rate of $\widetilde{O}(\nfrac{1}{\epsilon^2})$, under the assumptions which we list below.
          \begin{center}
            \fbox{
            \minipage{1.0\linewidth}
            Assumptions \hrule
            \vspace{2mm}
            \label{assumptions}
            $\mathrm{For\ all\ } i \in [n],\ j\in[m],  \mathrm{\ and \ } y_1,y_2\in \mathrm{supp}(f_{ij})$
            \vspace{1mm}

            \begin{enumerate}[noitemsep]
              \item \textcolor{white}{.}

              \vspace{-15mm}

              \begin{align*}
                &\hspace{-70mm}q_{ij} > \eta \hspace{70mm}\labelthis{$\eta$-coverage}\label{asmp:eta_coverage}
              \end{align*}
              \item\textcolor{white}{.}

              \vspace{-15mm}

              \begin{align*}
                \hspace{-33mm}\mu_{\min}\leq f_{ij}(y_1)\leq\mu_{\max} \hspace{33mm}\label{asmp:distributed_distribution}\labelthis{Distributed distribution}
              \end{align*}
              \item\textcolor{white}{.}

              \vspace{-15mm}

              \begin{align*}
                \hspace{-28.5mm}|f_{ij}(y_1)-f_{ij}(y_2)|< L|y_1-y_2|\hspace{+28.5mm}\label{asmp:lipschitz_distribution}\labelthis{Lipschitz distribution}
              \end{align*}
              \item \textcolor{white}{.}

              \vspace{-15mm}

              \begin{align*}
                \hspace{-32.0mm}\big|\eE[\phi_{ij}]\big| = \bigg|\ \int\limits_{\mathrm{supp}(f_{ij})}\hspace{-4mm} zf_{ij}(z)dz\ \bigg| < \rho.\labelthis{Bounded bid}\hspace{+34.5mm}\label{asmp:bounded_bid}
              \end{align*}
            \end{enumerate}
            \endminipage
            }
          \end{center}
          Assumption~(\hyperref[asmp:eta_coverage]{17}) guarantees that all advertisers have at least an $\eta$ probability of winning on every type,
          assumption~(\hyperref[asmp:distributed_distribution]{18}) places lower and upper bounds on the probability density functions of the $\phi_{ij}$,
          assumption~(\hyperref[asmp:lipschitz_distribution]{19}) guarantees that the probability density functions of the $\phi_{ij}$ are $L$-Lipschitz continuous,
          and assumption~(\hyperref[asmp:bounded_bid]{20}) assumes that the expected $\phi_{ij}$ is bounded.

          We expect Assumptions~(\hyperref[asmp:eta_coverage]{17}) and (\hyperref[asmp:bounded_bid]{20}) to hold in any real-world setting.
          We can drop the lower bound in Assumption~(\hyperref[asmp:distributed_distribution]{18}) by introducing ``jumps" in $\alpha$ to avoid ranges where the measure  of bids is small.
          Removing assumption~(\hyperref[asmp:lipschitz_distribution]{19}) would be an interesting direction for future work.
          \begin{remark}
            We inherit the assumption of \textit{independent} and \textit{regular distributions} from Myerson.
            In addition, we require the the distributions of valuations are \textit{strictly regular} to guarantee that ties between advertisers happen with $0$ probability.
            We can drop this assumption by incorporating a randomized tie-breaking rule which retains fairness.
            The above allocation rule is monotone and allocates the ad spot to the bidder with the highest shifted valuation $\phi_{ij}+\alpha_{ij}$ for a given user.
            Thus, it defines a unique truthful mechanism and corresponding payment rule.
          \end{remark}
          \section{Our Algorithm}
          \label{sec:algorithm}
          \begin{algorithm}[t!]
            \caption{\footnotesize Algorithm1($\cQ,G,L,\eta, \mu_{\max},\mu_{\min}, \epsilon$)}\label{Algorithm1}
            \textbf{Input:}
            Constraint polytope $\cQ\subseteq[0,1]^{n\times m}$,
            Lipschitz constant $G>0$ of $\mathrm{rev}(\cdot)$,
            Lipschitz constant $L>0$ of $f_{ij}(\cdot)$,
            minimum coverage $\eta>0$,
            lower and upper bounds, $\mu_{\min}$ and $\mu_{\max}$ of $f_{ij}(\cdot)$,
            and a constant $\epsilon>0$.\\
            \hspace{-8mm}\textbf{Output:} Shifts $\alpha\in\cR^{n\times m}$ for the optimal mechanism.

            \begin{algorithmic}[1]
              \STATE Initialize $\gamma\coloneqq\nfrac{\epsilon}{2G^2},\hspace{0.5mm} \xi\coloneqq (G\gamma)^2,\hspace{0.5mm} T=(\nfrac{\sqrt 2G}{\epsilon})^2$
              \STATE Compute $q_1\coloneqq\mathrm{proj}_{\cQ}(q(\mathit{0}_{n\times m}))$
              \STATE Compute $\alpha_{1}\coloneqq $ Algorithm2$(q_{t},\alpha_{t},\xi,L,\eta,\mu_{\max},\mu_{\min})$
              \FOR{t = 1,2,\dots,T}
              \STATE Compute $J_q(\alpha_{t})$
              \STATE Compute $\mathrm{rev}(q_{t})$ from $J_q(\alpha_{t})^\top\nabla \mathrm{rev}(q_{t})\coloneqq\nabla \mathrm{rev}_{\mathrm{shift}}(\alpha_{t})$
              \STATE Update $q_{t+1}\ \coloneqq\ \mathrm{proj}_{\cQ}(q_{t}+\gamma \nabla \mathrm{rev}(q_{t}))$
              \STATE  Update $\alpha_{t+1}\coloneqq $ Algorithm2$(q_{t},\alpha_{t},\xi,L,\eta,\mu_{\max},\mu_{\min})$
              \ENDFOR
              \STATE \bf return $\alpha$
            \end{algorithmic}
          \end{algorithm}
          Algorithm~\ref{Algorithm1} performs a projected gradient descent to find the optimal $q^\star\in\cQ$~\eqref{problem:convex_problem}.
          It starts with an initial coverage $q_1\in\cQ$, and the corresponding shift $\alpha_1=q^{-1}(q_1)$.
          At step $k$, it calculates the gradient $\nabla{\rm rev}(q_{k})$, by solving the linear-system in Eq.~(\hyperref[eq:linearSystem]{16}).
          To solve this linear-system, we need to calculate $J_{q}(\alpha_k)^\top$ and $\nabla {\rm rev}_{\rm shift}(\alpha_k)$.
          This can be done in $O(n^2m)$ steps if we have $\alpha_k=q^{-1}(q_k)$ (see Remark~\ref{rem:sparse_jacobian_2}).
          Therefore, the algorithm requires a ``good" approximation of $\alpha$ at each step, it maintains this by ``updating" the previous approximation $\alpha_{k-1}$
          using Algorithm~\ref{Algorithm2} to approximately solve the optimal-shift problem~\eqref{problem:nonconvex_unconstrained}.

          After calculating $\nabla{\rm rev}(q_{k})$, it takes a gradient step and projects the current iterate on $\cQ$ in $O((nm)^{\omega})$ time (Section~\ref{sec:proof_overview_projection_on_q}), where $\omega$ is the fast matrix multiplication coefficient.
          It takes roughly $O(\nfrac{1}{\epsilon^2})$ steps to obtain an $\epsilon$-accurate solution, and then returns its current shift $\alpha\approx\alpha^\star.$
          We can bound the error introduced by the approximation of $\alpha_k$ at each step by ensuring that Algorithm~\ref{Algorithm2} has sufficient accuracy.
          In particular, if it is $O(\epsilon^2)$ accurate we can prove that Algorithm~\ref{Algorithm1} converges in $\widetilde{O}(\nfrac{1}{\epsilon^2})$ steps.

          Next, we give the details of the projecting on $\cQ$ and calculating the gradient $\nabla {\rm rev}$.

          \subsection{Calculating and Bounding $\nabla {\rm rev(\cdot)}$}\label{sec:proof_overview_calculating_the_gradient}
          We fix the shift of one advertiser $i\in[n]$ for each type $j\in[m]$.
          Let $J_{q}(\alpha)$ be the Jacobian of the vectorized coverage, $\mathrm{vec}(q(\alpha))\in \cR^{(n-1)m}$, with respect to the vectorized shift, $\mathrm{vec}(\alpha)\in \cR^{(n-1)m}$.
          Then, $J_{q}(\alpha)$ is a $(n-1)m\times (n-1)m$ matrix
          \begin{align*}
            \hspace{5mm}J_q(\alpha)=\begin{bmatrix}
            \frac{\partial q_{11}(\alpha)}{\alpha_{11}} & \dots  & \frac{\partial q_{11}(\alpha)}{\alpha_{(n-1)1}} &\dots & \frac{\partial q_{11}(\alpha)}{\alpha_{(n-1)m}} \\
            \frac{\partial q_{21}(\alpha)}{\alpha_{11}} & \dots  & \frac{\partial q_{21}(\alpha)}{\alpha_{(n-1)1}} &\dots & \frac{\partial q_{21}(\alpha)}{\alpha_{(n-1)m}} \\
            \vdots & \vdots & \ddots & \vdots & \vdots \\
            \frac{\partial q_{(n-1)1}(\alpha)}{\alpha_{11}} & \dots  & \frac{\partial q_{(n-1)1}(\alpha)}{\alpha_{(n-1)1}} &\dots & \frac{\partial q_{(n-1)1}(\alpha)}{\alpha_{(n-1)m}}
          \end{bmatrix}.
          \end{align*}
          To obtain $\nabla \mathrm{rev}(q)$, we use the fact that $J_{q}(\alpha)$ is always invertible (Lemma~\ref{lem:Jacobian_is_invertible}).
          Given $\alpha = q^{-1}(\delta)$ for some $\delta\in[0,1]^{n\times m}$, we can calculate $\nabla {\rm rev}(\delta)$ by solving
          \begin{align*}
            \hspace{-3mm}\forall\ i\in[n],\ j \in [m],\ \frac{\partial \mathrm{rev}_{\mathrm{shift}}(\alpha)}{\partial \alpha_{ij}}&= \sum_{k\in[n]}\frac{\partial \mathrm{rev}(\alpha)}{\partial q_{kj}}\frac{\partial q_{kj}}{\alpha_{ij}}.
          \end{align*}
          Or equivalently by solving the linear-system in Eq. (\hyperref[eq:linearSystem]{16}).
          \begin{remark}
            \label{rem:fix_shift_for_inverting_jacobian}
            $J_q(\alpha)$ is invertible iff we fix the shift $\alpha_{ij}$ of one advertiser $i\in[n]$ for each type $j\in[m]$.
            Intuitively, if we increase the $\alpha_{ij}$ for all $i\in[n]$ and $j\in[m]$ by the same amount, then $q$ remains invariant.
            This implies that each row of $J_q(\alpha)$ has 0 sum, or that $J_q(\alpha)$ is not invertible.
          \end{remark}
          \begin{lemma}
            \textbf{(Jacobian is invertible).}
            \label{lem:Jacobian_is_invertible}
                For all $\alpha\in \cR^{(n-1)\times m}$, if
                all advertisers have non-zero coverage for all types $j\in [m]$,
                then $J_q(\alpha)\in \cR^{(n-1)\cdot m\times (n-1)\cdot m}$ is invertible.
          \end{lemma}
          \begin{proof}
            The coverage remains invariant if the bids of all advertisers are uniformly shifted for any given user type $j$. Therefore, for all $j\in [m]$ we have
            \begin{align}
              \sum_{t\in [n]} \frac{\partial q_{ij}}{\partial\alpha_{tj}} = 0.\label{eq:coverage_as_probability_1}
            \end{align}
            Since, increasing the shift $\alpha_{ij}$, does not increase the coverage $q_{kj}$ for any $k\neq i$, we have that
            \begin{align}
              \frac{\partial q_{kj}}{\partial\alpha_{ij}} \leq 0\label{eq:negative_gradient} {\rm \ and \ }\frac{\partial q_{ij}}{\partial\alpha_{ij}} \geq 0.
            \end{align}
            Now, from Equation~\eqref{eq:coverage_as_probability_1} we have
            \begin{align}
              \label{eq:coverageAsProbability}
              &\hspace{-4mm}\forall \ i \in [n],j \in [m], \ \frac{\partial q_{ij}}{\partial\alpha_{ij}} = \sum_{t\in [n]\backslash \{i\}} \bigg|\frac{\partial q_{ij}}{\partial\alpha_{tj}}\bigg|.
            \end{align}
            Further since the $n$-th advertiser has non-zero coverage, i.e., there is non-zero probability that advertiser $n$ bids higher than all other advertisers, changing $\alpha_{nj}$ must affect all other advertisers.
            In other words, for all $i\in[n-1]$ $\frac{\partial q_{ij}}{\partial\alpha_{nj}}\neq 0$.
            Using this we have,
            \begin{align}
              \forall \ i \in [n],j \in [m],  \ \frac{\partial q_{ij}}{\partial\alpha_{ij}} > \sum_{t\in [n-1]\backslash \{i\}} \bigg|\frac{\partial q_{ij}}{\partial\alpha_{tj}}\bigg|.
            \end{align}
            By observing that $q_{ij}$, on user type $j$, is independent of the $\alpha_{st}$, of any user type $t$ such that $t\neq j$, i.e.,
            \begin{align}
              \label{eq:Jacobian_is_sparse}
              \forall \ i,s \in [n],\ j,t \in [m],\ s.t.\ j\neq t,  \ \frac{\partial q_{ij}}{\partial\alpha_{st}} = 0,
            \end{align}
            and using Equation~\eqref{eq:coverageAsProbability}, we get that the Jacobian, $J_q(\alpha)$ is strictly diagonally dominant.
            Now, by the properties of strictly dominant matrices it is invertible.
          \end{proof}
          \begin{remark}\label{rem:sparse_jacobian_2}
            For all $i,s\in[n]$ such that $i\neq s$, $q_{ij}$ is independent of $\alpha_{st}$ \eqref{eq:Jacobian_is_sparse}.
            Therefore, that Jacobian $J_q(\alpha)$ is sparse.
            and the linear-system in Eq.~(\hyperref[eq:linearSystem]{16}) can be solved in $O(n^\omega m)$ steps, where $\omega$ is the fast matrix multiplication coefficient.
          \end{remark}

          \subsection{Projection on the Constraint Polytope ($\cQ$)}\label{sec:proof_overview_projection_on_q}
          Given any point $q\in [0,1]^{n\times m}$, by determining the constraints it violates, we can express the projection on the constraint polytope $\cQ$, as a quadratic program with equality constraints.
          Using this we can construct a projection oracle $\mathrm{proj}_{\cQ}$, which given a point $q\in [0,1]^{n\times m}$ projects it onto $\cQ$ in $O((nm)^\omega)$ arithmetic operations, where $\omega$ is the fast matrix multiplication coefficient.

          \section{Empirical Study}
          \begin{figure*}[t!]
            \hspace{-4mm}
            \subfigure[
            \footnotesize \label{fig:biasMeasure}
            \textit{Implicit Fairness of Keyword Pairs.}
            The x-axis depicts the fairness of the algorithm, measured by $(\ell,u)$-fairness constraints.
            We report number of auctions which satisfy each fairness level.
            We observe that 3282 auctions do not satisfy $\ell=0.3$ fairness constraint.
            ]
            {
            \vspace{-20mm}
            \includegraphics[trim=0.24cm 0.0cm 0.0cm 1.25cm,clip,height=4.8cm]{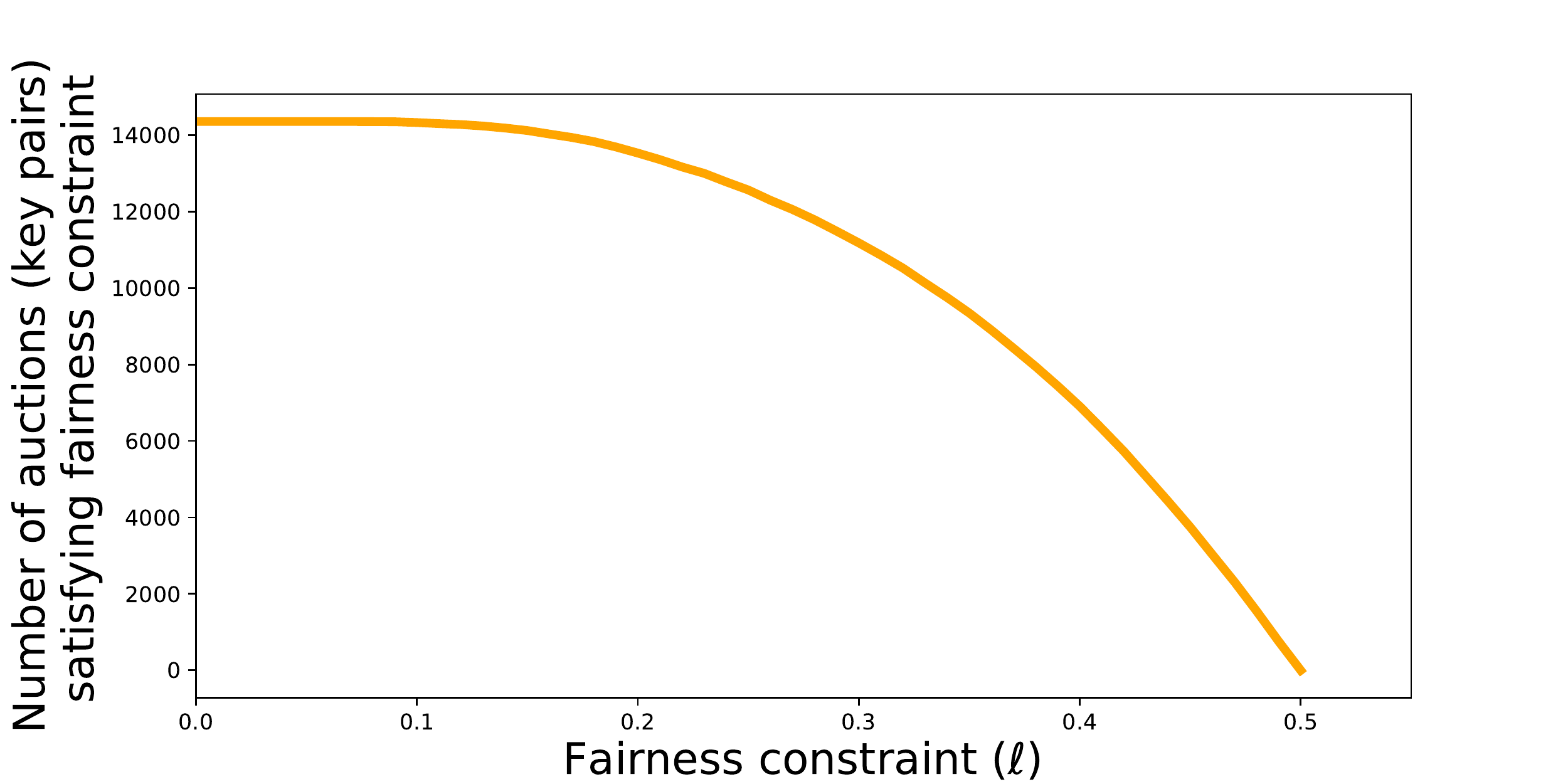}
            \hspace{-13mm}
            }
            \hspace{2mm}
            \subfigure[
            \label{fig:blockMatrix}
            \footnotesize
            \textit{Correlation Among Keywords.}
            The axes depict keywords, reordered to emphasize their correlation.
            A pair of keywords is colored white if it shares at least 2 advertisers.
            Each block can be interpreted as category of keyword (e.g., Science or Travel).
            ]
            {
            \hspace{2mm}
            \includegraphics[trim=1.5cm 0.3cm 2.2cm 0.9cm,clip,height=5.5cm]{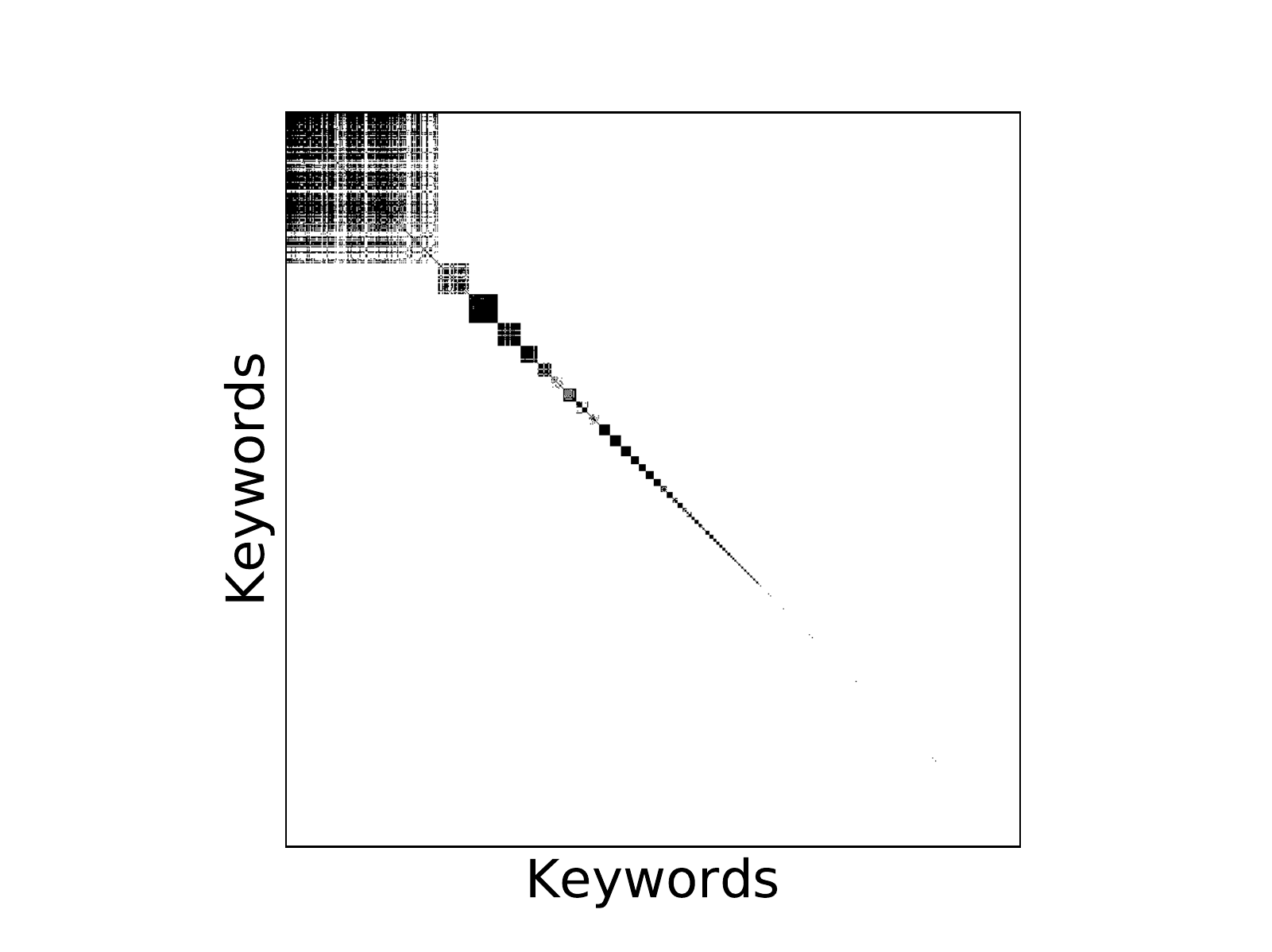}
            \hspace{2mm}
            }
            \caption{
            }
          \end{figure*}
          We evaluate our approach empirically on the {Yahoo!} A1 dataset~\cite{yah}.
          We vary the strength of the fairness constraint for all advertisers, find an optimal fair mechanism $\mathcal{F}$ using Algorithm~\ref{Algorithm1} and compare it against the optimal unconstrained (and hence potentially unfair) mechanism  $\mathcal{M}$, which is given by Myerson~\cite{myer}.
          We first consider the impact of the fairness constraints on the revenue of the platform.  Let  $\mathrm{rev}_\mathcal{N}$ denote the revenue of mechanism  $\mathcal{N}$. We report the revenue ratio
          $\kappa_{\mathcal{M},\mathcal{F}} \coloneqq \nfrac{\mathrm{rev}_\mathcal{F}}{\mathrm{rev}_\mathcal{M}}$.
          Note that the revenue of $\mathcal{F}$ can be at most that of $\mathcal{M}$, as it solves a constrained version of the same problem; thus $\kappa_{\mathcal{M},\mathcal{F}} \in [0,1]$.

          We then consider the impact of the fairness constraints on the advertisers.
          Towards this, we consider the distribution of winners among advertisers in an auction given by $\mathcal{M}$ and an auction given by $\mathcal{F}$. We  report the total variation distance $d_{TV}(\mathcal{M},\mathcal{F}) \coloneqq
          \nfrac{1}{2}\sum_{i=1}^{n}|\sum_{j=1}^{m}q_{ij}(\mathcal{M})-q_{ij}(\mathcal{F})| \in [0,1]$ between the two distributions, as a measure of how much the winning distribution changes due to the fairness constraints.

          Lastly, we consider the fairness of the resultant mechanism $\mathcal{F}$.
          To this end, we measure selection lift (${\rm slift}$) achieved by $\mathcal{F}$, ${\rm slift}(\mathcal{F}) \coloneqq \min_{i\in[n],j\in[m]}(\nfrac{q_{ij}}{1-q_{ij}})\in [0,1]$.
          Where ${\rm slift}(\mathcal{F})=1$, represents perfect fairness among the two user types.
          \subsection{Dataset}
          We use the Yahoo! A1 dataset~\cite{yah}, which contains bids placed by advertisers on the top 1000 keywords on \textit{Yahoo! Online Auctions} between June 15, 2002 and June 14, 2003.
          The dataset has 10475 advertisers, and each advertiser places bids on a subset of keywords;
          there are approximately $2\cdot10^{7}$ bids in the dataset.
          \begin{figure*}[t!]
            \centering
            \hspace{-1mm}\subfigure[
              \footnotesize \label{fig:slift}
              \textit{Fairness.}
              We report the fairness ${\rm slift}(\mathcal{F})$ achieved by our fair $(\mathcal{F})$ mechanism for varying level of fairness.
            ]
            {\includegraphics[trim=0cm 0cm 0cm 0cm,clip,height=5.9cm]{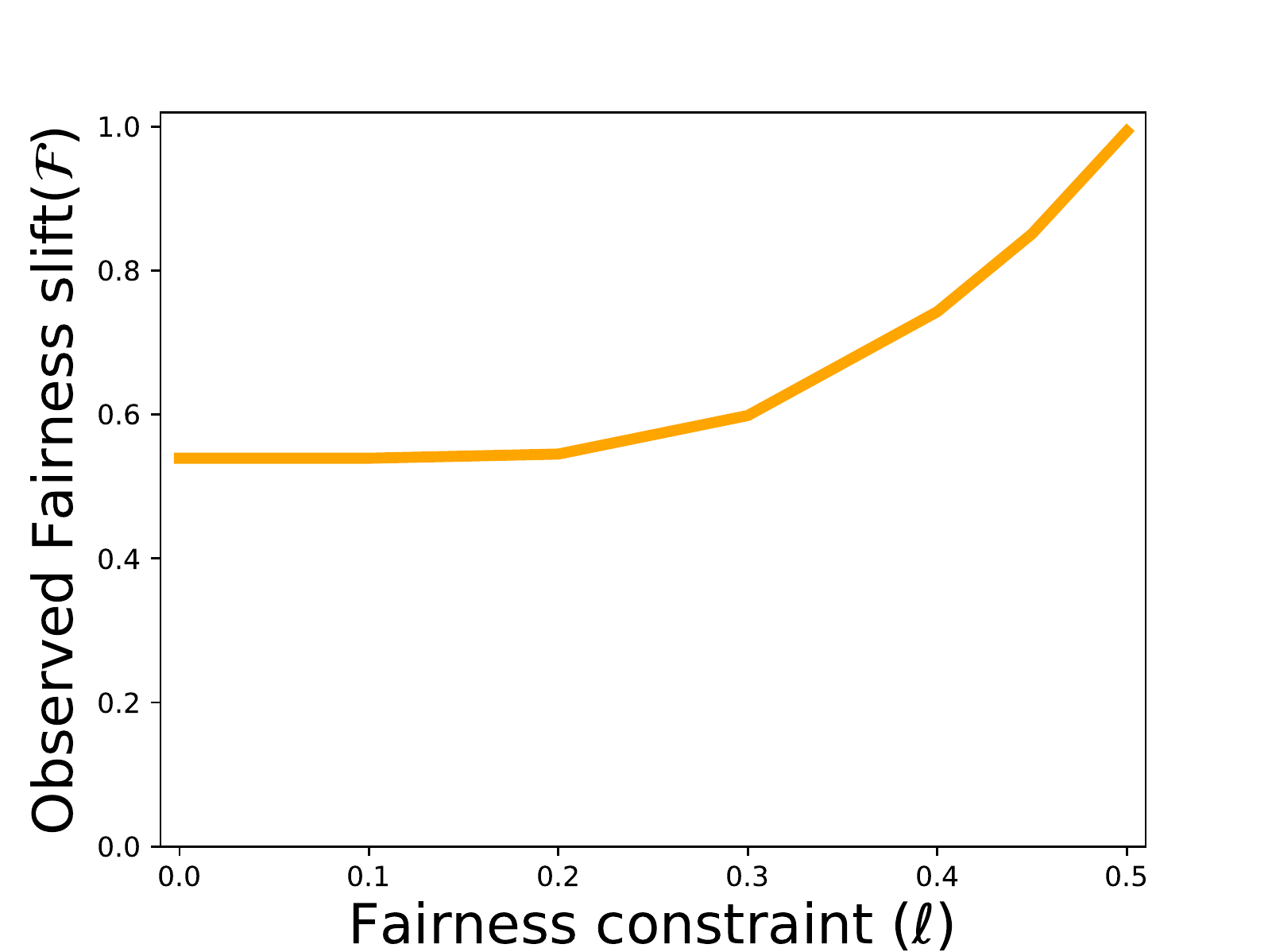}}
            \hspace{+1mm}
            \subfigure[
              \footnotesize\label{fig:rev_ratio}
              \textit{Fairness and Revenue.}
              We report the revenue ratio $\mathrm{rev}_{\mathcal{M},\mathcal{F}}$  between the fair $(\mathcal{F})$ and the unconstrained ($\mathcal{M}$) mechanisms.
            ]
            {
              \includegraphics[trim=0.0cm 0cm 0.1cm 1cm,clip,height=5.5cm]{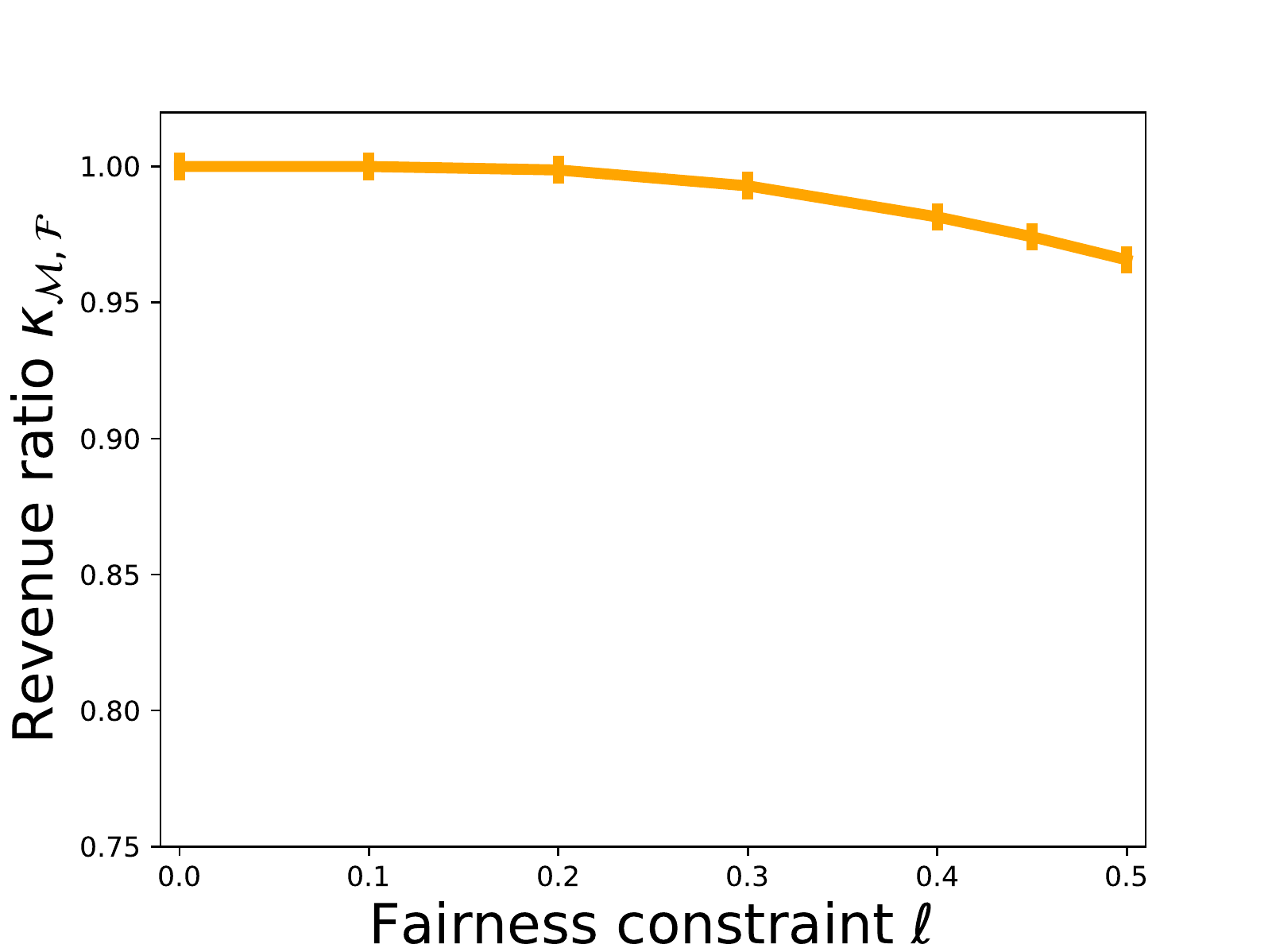}\hspace{-3mm}
            }
            \caption{
              \label{fig:Exp1Result1}
              The x-axis represents fairness constraint $\ell$ (lower bound).
              Error bars represent the standard error of the mean.
            }
          \end{figure*}

          For each keyword $k$, let $A_k$ be the set of advertisers that bid on it.
          We infer the distribution of valuation of an advertiser for a keyword by the bids they place on the keyword.
          In order to retain sufficiently rich valuation profiles for each advertiser, we remove advertisers who place less than 1000 bids on $k$ or whose valuations have variance lower than $3\cdot10^{-3}$ from $A_{k}$, and then those who win the auction less than $5\%$ of the time.
          This retains more than $1.5\cdot10^{7}$ bids.

          The actual keywords in the dataset are anonymized; hence, in order to determine whether two keywords  $k_1$ and $k_2$ are related, we consider whether they share more that one advertiser, i.e., $A_{k_1}\cap A_{k_2} > 1$.
          This allows us to identify keywords that are related (see Figure~\ref{fig:blockMatrix}), and hence for which spillover effects may be present as described in \cite{lambrecht_tucker}.
          Drawing that analogy, one can think of each keyword in the pair as a different type of user for which the same advertisers are competing, and the goal would be for the advertiser to win an equal proportion of each user.

          There are $14,380$ such pairs.
          However, we observe that spillover does not affect all keyword pairs (see Figure~\ref{fig:biasMeasure}).
          To test the effect of imposing fairness constrains in a challenging setting, we consider only the auctions which are not already fair; in particular there are $3282$ keyword pairs which are less than $\ell=0.3$ fair.

          \subsection{{Experimental Setup}}
          As we only consider pairs of keywords in this experiment, a lower bound constraint $\ell_{11} = \delta$ is equivalent to an upper bound constraint $u_{12}=1-\delta$.
          Hence, it suffices to consider lower bound constraints. We set $\ell_{i1}=\ell_{i2}=\ell \ \forall \ i \in [2]$, and vary $\ell$ uniformly from $0$ to $0.5$ , i.e., from the completely unconstrained case (which is equivalent to Myerson's action) to completely constrained case (which requires each advertiser to win each keywords in the pair with exactly the same probability).
          We report $\kappa_{\mathcal{N},\mathcal{M}}$, $d_{TV}(\mathcal{N},\mathcal{M})$, and ${\rm slift}(\mathcal{F})$ averaged over all auctions after $10^4$ iterations in Figure~\ref{fig:Exp1Result1}
           and Figure~\ref{fig:Exp1Result2};
          error bars represent the standard error of the mean over $10^4$ iterations and 3282 auctions respectively.
          \begin{remark}
            Computationally, we could consider more types ($m$). The bottleneck is empirical; whether the dataset contains enough keywords with $m$ overlapping advertisers for the experiment to be meaningful. For $m < 7$ we get over 1000 such keywords sets, and observe results similar to $m = 2$ case, losing less than $5\%$ of the revenue with a TV-distance smaller than 0.05 even for the setting with $\ell=0.5$.
          \end{remark}
          \begin{figure*}[t!]
            \centering
            {
              \includegraphics[trim=-0.2cm 0cm -0.2cm 1cm,clip,height=5.5cm]{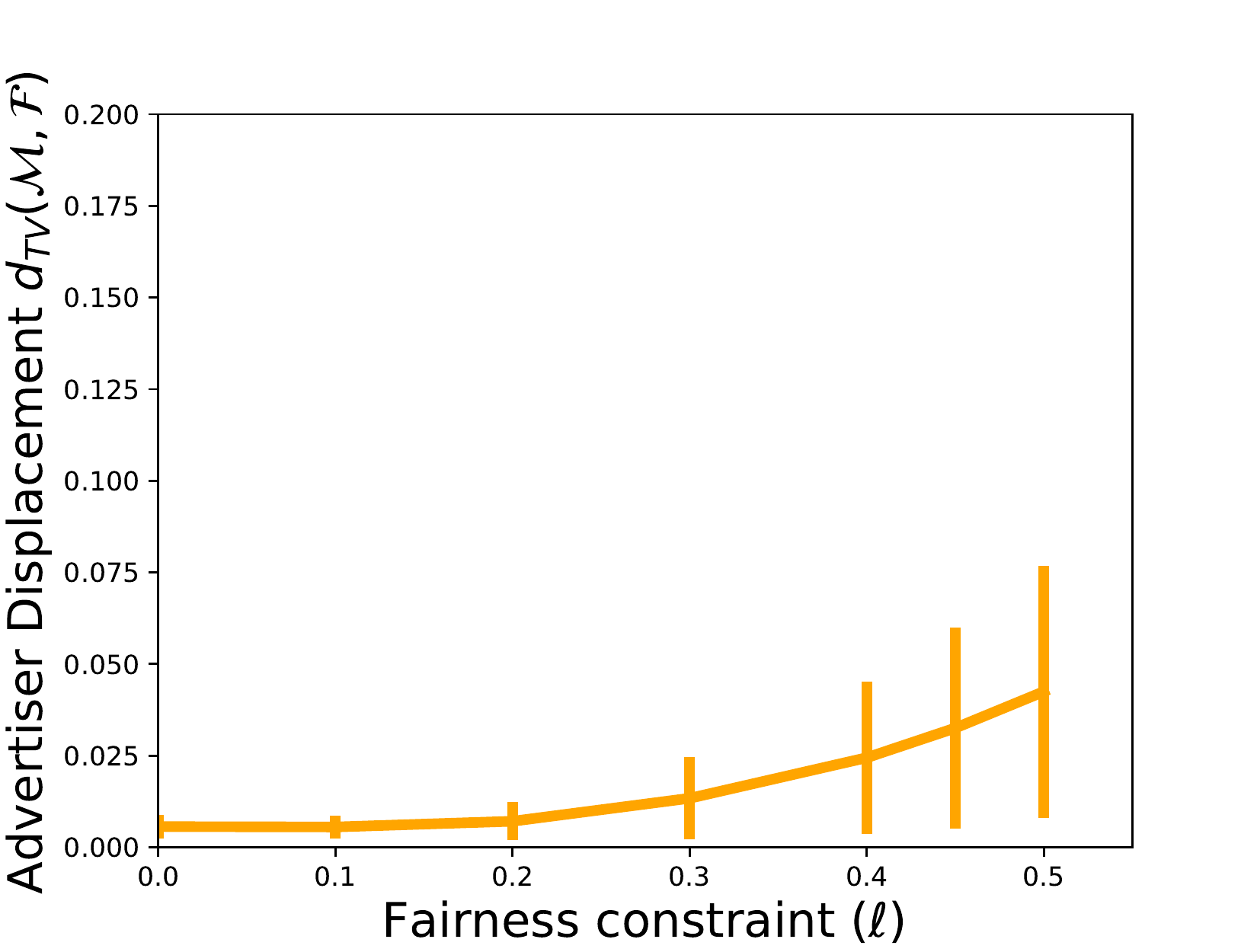}
            }
            \caption{
              \label{fig:Exp1Result2}
              \textit{Effect of Fairness on Advertisers.}
              The x-axis represents fairness constraint $\ell$ (lower bound).
              We report the total variation distance, $d_{TV}(\mathcal{M},\mathcal{F})$, between the distribution of winners in ads allocated by the fair $(\mathcal{F})$ and the unconstrained $(\mathcal{M})$ mechanism.
              Error bars represent the standard error of the mean.
            }
          \end{figure*}

          \subsection{Empirical Results}
          \noindent{\em Fairness.}
          Since the auctions are unbalanced to begin with, we expect the selection lift to increase with the fairness constraint.
          We observe a growing trend in the selection lift, eventually achieving perfect fairness for $\ell=0.5$.\\

          \noindent{\em Revenue Ratio.}
          We do not expect to outperform the optimal unconstrained mechanism.
          However, we observe that even in the perfectly balanced setting with $\ell=0.5$ our mechanisms lose less than  5\% of the revenue.\\

          \noindent{\em Advertiser Displacement.}
          Since the auctions are unbalanced to begin with, we expect TV-distance to grow with the fairness constraint.
          We observe this growing trend in the TV-distance on lowering the risk-difference.
          Even for zero risk-difference ($\ell=0.5$) our mechanisms obtain a TV-distance smaller than $0.05$.
          We present a discussion of this result in Section~\ref{sec:small_tv}.
          %


          \section{Proofs}

          \subsection{Proof of Theorem~\ref{thm:1}}
          \label{sec:proof1}
          \begin{proof}
            Let us introduce three Lagrangian multipliers, a vector $\alpha_j\in\cR_{\geq 0}^n$,
            a vector $\beta_j\in\cR_{\geq 0}^{n} \ \mathrm{and} $
            a continuous function $\gamma_j(\cdot )\colon {\rm supp}(\phi_j)\to\cR_{\geq 0}\ \forall \ j \in [m]$,
            for the lower bound, upper bound, and single item constraints respectively.
            Then calculating the Lagrangian function we have
            \begin{align*}
              &\hspace{-10mm}\text{\footnotesize{$L \coloneqq
              \hspace{-3mm}\hspace{2.5mm}
                \sum_{j\in[m]}\hspace{1mm} \Pr_{\cU}[j]\sum_{i\in[n]}\hspace{1mm}\int_{{\rm supp}(\phi_j)}\hspace{-9mm}\phi_{ij}x_{ij}(\phi_j) df_j(\phi_j)
              \hspace{7mm}+
                \sum_{\substack{i\in[n]\\j\in[m]}}\alpha_{ij} \bigg( \int_{{\rm supp}(\phi_j)}\hspace{-9mm} x_{ij}(\phi_{ij}) df_j(\phi_j)
               -
                  \ell_{ij}\hspace{-2mm} \sum_{t\in[m]}\int_{{\rm supp}(\phi_t)}\hspace{-9mm} x_{it}(\phi_t) df_t(\phi_t)
              \bigg)$}}\\
              &\hspace{-10mm}\hspace{-3mm}\quad\quad\quad\
              \text{\footnotesize{$
              \hspace{-6mm}+\sum_{j\in[m]}
                \hspace{1mm}\int_{{\rm supp}(\phi_j)}\hspace{-9mm}\gamma_j(\phi_j)\big(1-\sum_{i\in[n]}x_{ij}(\phi_{ij})\big)df_j(\phi_j)
              \hspace{1mm}- \sum_{\substack{i\in[n]\\j\in[m]}}\beta_{ij}
                \bigg(
                  \int_{{\rm supp}(\phi_j)}\hspace{-9mm} x_{ij}(\phi_{ij})df_j(\phi_j)
                -
                  u_{ij}\hspace{-2mm} \sum_{t\in[m]}\hspace{-1mm}\int_{{\rm supp}(\phi_t)}\hspace{-9mm} x_{it}(\phi_{it})df_t(\phi_t)
                 \bigg)$}}.
            \end{align*}
            The second integral is well defined by from the continuity of $\gamma_j(\cdot)$ and monotonic nature of $x_j(\cdot)$.
            In order for the supremum of the Lagrangian over $x_{ij}(\cdot)\geq 0$ to be bounded, the coefficient of $x_{ij}(\cdot)$ must be non-positive. Therefore we require that for all $g \subseteq {\rm supp}(\phi_j),\ i \in [n],$ and $j\in [m]$
            \begin{align*}
              \int\limits_{g}\hspace{-1mm}
              \alpha_{ij}-\beta_{ij}+\Pr_{\cU}[j]\phi_{ij}
              -\hspace{-1.5mm}\sum_{t\in[m]}(\alpha_{it}\ell_{it}-\beta_{it} u_{it})
               - \gamma_j(\phi_j) df_j(\phi_j)\leq 0.
            \end{align*}
            Since $x_{ij}(\cdot)$ and $\gamma_j(\cdot)$ are continuous, we can equivalently require for all $\phi_j,\ i \in [n],$ and $j \in [m]$
            \begin{align*}
              &\hspace{-10mm}\alpha_{ij}-\beta_{ij}+\Pr_{\cU}[j]\phi_{ij}
              -\sum_{t\in[m]}(\alpha_{it}\ell_{it}-\beta_{it} u_{it})
               - \gamma_j(\phi_j) \leq 0.
            \end{align*}
            If this holds, we can express the supremum of $L$ as
            \begin{align*}
              &\hspace{-7mm}\sup_{x_{ij}(\cdot)\geq0} L =
              \sum_{j\in[m]}\int_{{\rm supp}(\phi_j)}\gamma_j(\phi_j)df_j(\phi_j).
            \end{align*}
           Now we can express the \textit{dual optimization problem} as follows:

           \vspace{2mm}
            \noindent{\bf (Dual of the infinite-dimensional fair advertising problem\label{eq:dual}).} For all $j\in[m]$,
            find a optimal $\alpha_j\in\cR_{\geq 0}^{n},\ \beta_j\in\cR_{\geq 0}^{n}$ and $\gamma_j(\cdot )\colon {\rm supp}(\phi_j) \to\cR_{\geq 0}$ for
                \begin{align*}
                  &\hspace{3mm}\min_{\substack{\alpha_{j}\geq 0\\ \beta_{j}\geq 0\\\hspace{-3mm} \gamma_{j}(\cdot) \geq 0}}   \quad \sum_{j\in[m]}\int_{{\rm supp}(\phi_j)}\gamma_j(\phi_j) df_j(\phi_j)\numberthis\\
                  &\hspace{3mm}\hspace{3mm}\mathrm{s.t.}\quad \ \alpha_{ij}-\beta_{ij}+\Pr_{\cU}[j]\phi_{ij}
                  -\hspace{-1mm}\sum_{t\in[m]}(\alpha_{it}\ell_{it}-\beta_{it} u_{it})
                  \leq  \gamma_j(\phi_j)\ \ \forall\ i \in [n], j \in [m], \phi_{j}.\numberthis
                \end{align*}
            Since the \hyperref[eq:primal]{primal} is linear in $x_{ij}(\cdot)$, and the constraints are feasible strong duality holds.
            Therefore, the dual optimal is primal optimal.

            For any feasible constraints we have for all $i\in [n]\ $ $\sum_{j\in [m]}\ell_{ij} \leq 1$ and $\sum_{j\in [m]}u_{ij} \geq 1$.
            Therefore the coefficient of $\alpha_{ij}$, $1-\sum_{j\in [m]}\ell_{ij} \geq 0$, and that of $\beta_{ij}$, $\sum_{j\in [m]}u_{ij}-1\geq 0$. Since $\alpha$ and $\beta$ are non-negative, a optimal solution to the dual is finite.
            Let $\alpha^\star,\beta^\star$ 
            be a optimal solutions to the \hyperref[eq:dual]{dual}, and $x_{ij}^\star(\cdot)$ be a optimal solution to the \hyperref[eq:primal]{primal}.
            Fixing $\alpha$ and $\beta$ to their optimal values $\alpha^\star$ and $\beta^\star$ in the \hyperref[eq:dual]{dual}, let us define new virtual valuations $\phi_{ij}^\prime$, for all $i\in [n]$ and $j\in [m]$
            \begin{align*}
              \hspace{10mm}\phi_{ij}^\prime \coloneqq \phi_{ij} + \frac{1}{\Pr_{\cU}[j]}\big( \alpha_{ij}^\star -\beta_{ij}^\star - \sum_{t\in[m]}(\alpha_{it}^\star \ell_{it} - \beta_{it}^\star u_{it}) \big).
            \end{align*}
            Then the leftover problem has only one Lagrangian multiplier, $\gamma_j(\cdot)$. Let $\gamma_{j}^\prime(\cdot)$ be the affine transformation of $\gamma_j$ defined on virtual valuations, i.e., $\gamma_j^\prime(\phi_j^\prime)\coloneqq\gamma_j(\phi_j)$, then the problem can be expressed as follows.
              \noindent {\bf (Dual with shifted virtual valuations).} For all $j\in[m]$, find the optimal $\gamma_j^\prime(\cdot )\colon {\rm supp}(\phi_j^\prime) \to\cR_{\geq 0}$ for
                \begin{align*}
                  &\hspace{-30mm}\min_{ \gamma_{j}(\cdot) \geq 0} \quad  \sum_{j\in[m]}\int_{{\rm supp}(\phi_j)}\gamma_j(\phi_j^\prime) df_j(\phi_j^\prime)\numberthis\\
                  &\hspace{-29mm}\hspace{3mm}\mathrm{s.t.} \hspace{2mm}\quad \Pr_{\cU}[j]\phi_{ij}^\prime
                  \leq  \gamma_j(\phi_j^\prime)
                  \quad \forall\ i \in [n], j \in [m], \phi^\prime.\numberthis
                \end{align*}
            \noindent This is the dual of the following unconstrained revenue maximizing problem.
            Myerson's mechanism is the revenue maximizing solution to the unconstrained optimization problem.
            Further, by linearity and feasibility of constraints strong duality holds.
            Therefore the $\alpha^\prime$-shifted mechanism, for $\alpha^\prime=\nfrac{1}{\Pr_{\cU}[j]}\cdot\big( \alpha_{ij}^\star
            -\beta_{ij}^\star+\sum_{t\in[m]}(\alpha_{it}^\star \ell_{it} - \beta_{it}^\star u_{it}) \big)$ is a optimal fair mechanism.

            \vspace{2mm}
            \noindent {\bf (Unconstrained primal for the infinite-dimensional fair advertising problem).} For all $j\in[m]$, find the optimal allocation rule $x_j(\cdot)\colon\cR^n\to[0,1]^n$ for
              \begin{align*}
                &\hspace{-29mm}\max_{x_{ij}(\cdot)\geq 0} \mathrm{rev}_{\mathcal{M}}(x_1,x_2,\dots,x_m)
                \\
                &\hspace{-27mm}\hspace{3mm}\mathrm{s.t.} \hspace{1mm}
                    \sum_{i\in[n]}x_{ij}(\phi_j) \leq 1  \quad \forall \ j \in [m], \phi_j \in{\rm supp}(\phi_j).
              \end{align*}

            \noindent Further, Myerson's mechanism is truthful if the distribution of valuations are regular and independent.
            Since $\alpha$-shifted mechanism applies a constant shift to all valuation, it follows under the same assumptions that any $\alpha$-shifted mechanism is also truthful, and therefore has a unique payment rule defined by its allocation rule.
            \end{proof}

          \subsection{Proof of Theorem~\ref{thm:2}}
            \label{sec:proof_of_thm:2}
            \paragraph{Supporting Lemmas.}
            Towards the proof of Theorem~\ref{thm:2} we require the following two Lemmas.
            The first lemma shows that ${\rm rev}(\cdot)$ is Lipschitz continuous.
            Its proof is presented in Section~\ref{app:proof_of_revenue_is_Lipschitz}.
            \begin{lemma}
              \textbf{(Revenue is Lipschitz).}
              \label{lem:revenue_is_Lipschitz}
              For all coverages $q_1,q_2$ $\in\cQ$,
              if assumptions~(\hyperref[asmp:eta_coverage]{17}), (\hyperref[asmp:distributed_distribution]{18}) and (\hyperref[asmp:bounded_bid]{20}) are satisfied,
              then
              \begin{align}
                \hspace{-3mm}|\mathrm{rev}(q_1)-\mathrm{rev}(q_2)|\leq \bigg(\frac{\mu_{\max}\rho}{\mu_{\min}\eta}\bigg) n^2\|q_1-q_2\|_F\footnotemark.
              \end{align}
              \footnotetext{We use $\|\cdot\|_F$ to denote the Frobenius norm.}
            \end{lemma}
            \noindent The next lemma is an algorithm to solve the {\em optimal shift problem}.
            Its proof is presented in Section~\ref{app:proof_of_thm:3}
            \begin{lemma}\label{thm:3}
              \textbf{(An algorithm to solve the optimal shift problem).}
              There is an algorithm (Algorithm~\ref{Algorithm2}) which outputs $\alpha\in\cR^{n\times m}$ such that
              if assumptions~(\hyperref[asmp:eta_coverage]{17}), (\hyperref[asmp:distributed_distribution]{18}) and (\hyperref[asmp:lipschitz_distribution]{19}) are satisfied,
              then $\alpha$ is an $\epsilon$-optimal solution for the optimal shift problem, i.e., $\cL(\alpha)<\epsilon$, in
              \begin{align*}
                \log\bigg({\frac{m\cL(\alpha_1)}{\epsilon}}\bigg)\frac{n^3(L+n^2\mu_{\max}^2)}{(\eta\mu_{\min})^2} \ \mathrm{steps.}
              \end{align*}
              Where the arithmetic operations in each step are bounded by calculating $\nabla \cL$ once.
            \end{lemma}
            \begin{proof}[Proof of Theorem~\ref{thm:2}]
              Starting from $q_0\in \cQ$, Algorithm~\ref{Algorithm1} performs a projected gradient descent on $\cQ$.
              Since $\cQ$ is convex, the projection is contractive. In particular, for the optimal $q^\star\in\cQ$
              \begin{align*}
                \forall \ q, \ \|\mathrm{proj}_{\cQ}(q)-q^\star\|_2\leq \|q-q^\star\|_2 \numberthis\label{eq:contracting_projection}.
              \end{align*}
              It queries the shift $\alpha_k\approx q^{-1}(q_{k})$ from Algorithm~\ref{Algorithm2} at each step.
              This introduces some error $\xi>0$ at each step, which we fix later in the proof.

              Let $z_{k+1}=q_{k}+\gamma\nabla \mathrm{rev}(q_k)$ be the coverage at the $k+1$-th gradient-step, and $q_{k+1}=q(\alpha_{k+1})$ be the coverage obtained by  querying
              $\alpha_k\approx q^{-1}(\mathrm{proj}_{\cQ}(q_{k+1}))$.
              Then, we have the following bound on the error
                \begin{align}
                  &\hspace{42mm}\|\mathrm{proj}_{\cQ}(z_{k+1}\hspace{-0.5mm})-q_{k+1}\|_2^2 \leq \xi. \label{eq:error_from_algorithm_2}
                  \labelthis{Error from Algorithm~\ref{Algorithm2}}
                \end{align}
              We know that $\mathrm{rev}(\cdot)$ is a concave function of $q$.
              Using the first-order condition of concavity at $q^\star$ and $q_k$ we have
              \begin{align*}
                &\hspace{-5mm}\hspace{15mm}\|z_{k+1}-q^\star\|_2^2 = ||q_{k}+\gamma\nabla \mathrm{rev}(q_k)-q^\star\|_2^2 \\
                &\hspace{-5mm}\hspace{15mm}\hspace{21mm}\stackrel{}{\leq}
                \|q_{k}-q^\star\|_2^2 + 2\gamma(\mathrm{rev}(q_k -\mathrm{rev}(q^\star) + \gamma^2 \|\nabla \mathrm{rev}(q_k)\|_2^2.\numberthis \label{eq:contraction}
              \end{align*}
              Using the triangle inequality with Eq.~\eqref{eq:contracting_projection} and (\hyperref[eq:error_from_algorithm_2]{32}) we get
              \begin{align*}
                \|q_{k+1}-q^\star\|_2^2 &= \|q_{k+1}-\mathrm{proj}_{\cQ}(z_{k+1})+\mathrm{proj}_{\cQ}(z_{k+1})-q^\star\|_2^2\\
                &\leq \|z_{k+1}-q^\star\|_2^2+\xi\label{eq:erroriteration}\numberthis\\
                &\hspace{-0.8mm}\stackrel{\eqref{eq:contraction}}{\leq} \|q_{k}\text{--}q^\star\|_2^2 + 2\gamma(\mathrm{rev}(q_k\hspace{-0.4mm})\text{--} \mathrm{rev}(q^\star\hspace{-0.2mm})\hspace{-0.3mm}) + \gamma^2 \|\nabla \mathrm{rev}(q_k)\hspace{-0.2mm}\|_2^2 + \xi\hspace{-6mm}\numberthis\label{eq:recurrence}
              \end{align*}
              Expanding the above recurrence we get
              \begin{align*}
                &\hspace{-3mm}\|q_{k+1}-q^\star\|_2^2\stackrel{\eqref{eq:recurrence}}{\leq} k\xi + \|q_{1}-q^\star\|_2^2+\sum\nolimits_{i=1}^{k}\gamma^2\|\nabla \mathrm{rev}(q_i)\|_2^2 \\
                &\quad\quad\quad\quad\quad\quad\hspace{3.1mm} +2\sum\nolimits_{i=1}^{k}\gamma(\mathrm{rev}(q_i) - \mathrm{rev}(q^\star)).
                \numberthis\label{eq:expanded_recurrence}
              \end{align*}
              Substituting $\|q_{k+1}-q^\star\|_2^2 \geq 0$, and $\|q_{1}-q^\star\|_2^2 \leq  1$ we get
              \begin{align*}
                k\xi + 1 + 2\sum_{i\in[k]}\gamma(\mathrm{rev}(q_i) - \mathrm{rev}(q^\star)) + \sum_{i\in[k]}\hspace{-1.5mm}\gamma^2 \|\nabla \mathrm{rev}(q_i)\|_2^2\geq 0.
              \end{align*}
              Replacing $\mathrm{rev}(q_i)$ by its maximum,
              choosing $\xi\coloneqq G^2\gamma^2$,
              and using $\|\nabla \mathrm{rev}(q_i)\|_2 \leq G$
              and $k\coloneqq \big(\nfrac{\sqrt 2 G}{\epsilon}\big)^2$ we get
              \begin{align*}
                \mathrm{rev}(q^\star)- \max_{i\in [k]}\big(\mathrm{rev}(x_i)\big) &\leq \frac{1 + k\xi + G^2\sum_{i\in[k]}\gamma^2}{2\sum_{i\in[k]}\gamma}
                \leq \epsilon
              \end{align*}
              At each step we perform a small update to $q_k$ and query $\alpha_k$, therefore, Algorithm~\ref{Algorithm2} is always warm-started, i.e., $\|z_{k+1}-q_{k}\|_2^2< G\gamma$.
              Now, from Lemma~\ref{thm:3} the total steps required to update $\alpha$ are
              \begin{align*}
                &\sum\nolimits_{i=1}^{k}\log\big({\nfrac{m G\gamma}{\xi}}\big)\cdot{n^3(L+n^2\mu_{\max}^2)}{(\eta\mu_{\min})^{-2}}\\
                &= (\nfrac{\sqrt 2 G}{\epsilon})^2
                \log\big({\nfrac{2m G}{\epsilon}}\big)\cdot{n^3(L+n^2\mu_{\max}^2)}{(\eta\mu_{\min})^{-2}}
              \end{align*}
              The sum of the total gradient steps by Algorithm~\ref{Algorithm1}, and the total gradient steps by all calls of Algorithm~\ref{Algorithm2} is
              \begin{align*}
                O\big(\nfrac{G^2}{\epsilon^2}
                \log\big({\nfrac{2m G}{\epsilon}}\big)\cdot{n^3(L+n^2\mu_{\max}^2)}{(\eta\mu_{\min})^{-2}}\big).
              \end{align*}
              Using $G=(\nfrac{\mu_{\max}\rho}{\mu_{\min}\eta})\cdot n^2$ (from Lemma~\ref{lem:revenue_is_Lipschitz}) we have that Algorithm~\ref{Algorithm1} gets an $\epsilon$-approximation of optimal revenue in
              \begin{align*}
                \widetilde{O}\bigg(\frac{n^7\log m}{\epsilon^2} \frac{(\mu_{\max}\rho)^2}{(\mu_{\min}\eta)^4}
                (L+n^2\mu_{\max}^2)\bigg) \ \mathrm{steps.}
              \end{align*}
              Where $\widetilde{O}$ hides $\log$ factors in $n, \rho,\eta,\mu_{\max}, \nfrac{1}{\epsilon}$ and  $\nfrac{1}{\mu_{\min}}$.
            \end{proof}
          \subsection{Proof of
          Lemma~\ref{lem:revenue_is_Lipschitz}}
            \label{app:proof_of_revenue_is_Lipschitz}
            We use Lemma~\ref{lem:rev_is_Lipschitz_in_shifts} and Lemma~\ref{lem:alpha_Lipschitz_in_coverage} in the proof of Lemma~\ref{lem:revenue_is_Lipschitz}.
            The two lemmas split the Lipschitz continuity of ${\rm rev}(\cdot)$ into the Lipschitz continuity of ${\rm rev}_{\rm shift}(\cdot)$ and $\alpha_{ij}=q^{-1}_{ij}(\cdot)$ respectively.
            Their proofs are follow in Section~\ref{app:proof_of_rev_is_Lipschitz_in_shifts} and Section~\ref{app:proof_of_alpha_Lipschitz_in_coverage} respectively.
            \begin{lemma}\label{lem:rev_is_Lipschitz_in_shifts}
              \textbf{(Revenue is Lipschitz continuous in shifts).}
              For all $\alpha \in \cR^{(n-1)\times m}$,
              if $\mathrm{pdf}$, $f_{ij}(\phi)$ of the virtual valuations is bounded above by $\mu_{\max}$, and $\phi_{ij}$ is bounded above by $\rho$ $ \ \forall \ i\in[n],\ j\in[m]$,
              then $\mathrm{rev}_{\mathrm{shift}}(\alpha)$ is $(\mu_{\max}\rho n^{\frac{3}{2}})$-Lipschitz continuous.
            \end{lemma}
            \begin{lemma}
              \label{lem:alpha_Lipschitz_in_coverage}
              \textbf{(Shifts is Lipschitz continuous in coverage).}
              For all $\alpha,\beta\in \cR^{(n-1)\times m}$,
              such that $q_{ij}(\beta+t(\alpha-\beta)) > \eta$,
              if the probability density function, $f_{ij}(\cdot)$, of virtual valuations is bounded by $\mu_{\min}$ and $\mu_{\max}$
              $\ \forall \ t\in [0,1],\ i\in[n],\ j\in[m]$,
              then
              \begin{align*}
                \hspace{8mm}\hspace{1mm}\|\alpha-\beta\|_F < \frac{\sqrt n}{\eta\mu_{\min}}\|q(\alpha)-q(\beta)\|_2.
              \end{align*}
            \end{lemma}
            \begin{proof}[\unskip\nopunct]
              \textit{Proof of Lemma~\ref{lem:revenue_is_Lipschitz}.}
              Let $\alpha$, $\beta\in \cR^{(n-1)\times m}$ be the shifts achieving $q_1$ and $q_2$ respectively. Then by Lemma~\ref{lem:rev_is_Lipschitz_in_shifts} and Lemma~\ref{lem:alpha_Lipschitz_in_coverage} we have
              \begin{align}
                \label{eq:revenue_Lipschitz_eq1}
                \hspace{8mm}\hspace{-21mm}|\mathrm{rev}(q(\alpha))-\mathrm{rev}(q(\beta))|&\stackrel{{\rm Lemma}~\ref{lem:rev_is_Lipschitz_in_shifts}}{\leq} \mu_{\max}\rho n^{\frac{3}{2}} \|\alpha-\beta\|_{F}\\
                \hspace{8mm}\hspace{-21mm}\label{eq:revenue_Lipschitz_eq2}
                \|\alpha-\beta\|_F &\stackrel{{\rm Lemma}~\ref{lem:alpha_Lipschitz_in_coverage}}{<} \frac{\sqrt n}{\eta\mu_{\min}}\|q(\alpha)-q(\beta)\|_2.
              \end{align}
              By combining Equation~\eqref{eq:revenue_Lipschitz_eq1} and Equation~\eqref{eq:revenue_Lipschitz_eq2} we get the required result
              \begin{align}
                |\mathrm{rev}(q_1)-\mathrm{rev}(q_2)|\stackrel{\eqref{eq:revenue_Lipschitz_eq1},\eqref{eq:revenue_Lipschitz_eq2}}{<}
                \frac{\mu_{\max}\rho}{\mu_{\min}\eta} n^2\|q_1-q_2\|_2.\label{eq:lem:revenue_is_Lipschitz_eq1}
              \end{align}
            \end{proof}

            \subsubsection{Proof of Lemma~\ref{lem:rev_is_Lipschitz_in_shifts}}
              \label{app:proof_of_rev_is_Lipschitz_in_shifts}
              \begin{proof}
                We first consider the revenue for one user type $j$, $\mathrm{rev}_{\mathrm{shift},\ j}(\alpha)$, and then combine the result across all user type to show that $\mathrm{rev}_{\mathrm{shift}}(\alpha)$ is Lipschitz continuous.
                Formally, we define $\mathrm{rev}_{\mathrm{shift},\ j}(\alpha)$ as
                \begin{align*}
                  \mathrm{rev}_{\mathrm{shift},\ j}(\alpha)&\coloneqq \sum_{i\in[n]}
                  \Pr_{\cU}[j]\hspace{-3mm}\int\limits_{\mathrm{supp}(f_{ij})}\hspace{-2mm}y f_{ij}(y)\hspace{-3mm}
                  \prod_{k\in [n]\backslash \{i\}}\hspace{-4mm}F_{kj}(y+\alpha_{ij}-\alpha_{kj})dy.
                  \labelthis{Revenue from type $j$}
                \end{align*}
                Then the total revenue $\mathrm{rev}_{\mathrm{shift}}(\alpha)$ is just a sum of $\mathrm{rev}_{\mathrm{shift},\ j}(\alpha)$ for all user types
                \begin{align*}
                  \mathrm{rev}_{\mathrm{shift}}(\alpha)&=\sum_{j=1}^{m}\mathrm{rev}_{\mathrm{shift},\ j}(\alpha).
                \end{align*}
                We can express $\nabla \mathrm{rev}_{\mathrm{shift},\ j}(\alpha)$ as shown in Figure~\ref{fig:missing_equations_from_rev_is_Lipschitz}.
                \begin{figure*}[t!]
                  \begin{center}
                    \fbox{
                    \minipage{\linewidth}
                    \footnotesize{
                    \begin{flalign*}
                      &{\rm For\ all\ }j,k\in [m],\  i\in [n-1], \st j\neq k &\text{\white{.}}\\
                      &\quad\frac{\partial \mathrm{rev}_{\mathrm{shift},\ j}(\alpha)}{\partial \alpha _{ij}} =
                      \Pr_{\cU}[j] \sum _{k\neq i}\hspace{1mm}\int\limits_{\mathrm{supp}(f_{ij})}\hspace{-5mm}y f_{ij}(y) f_{kj}(y +\alpha _{ij} -\alpha _{kj})\prod _{\ell \neq i,k}\hspace{-1mm}F_{\ell j}(y +\alpha _{ij} -\alpha _{\ell j}) dy \numberthis\label{eq:Lem67eq1}\\
                      &\quad\hspace{17mm}\quad -\Pr_{\cU}[j]\sum _{k\neq i}\hspace{1mm}\int\limits_{\mathrm{supp}
                      (f_{kj})}\hspace{-5mm}y f_{kj}(y ) f_{ij}(y +\alpha _{kj} -\alpha _{ij})\prod _{\ell \neq i,k} F_{\ell j}(y +\alpha _{kj} -\alpha _{\ell j})dy\\
                      &\quad\frac{\partial \mathrm{rev}_{\mathrm{shift},\ j}(\alpha)}{\partial \alpha _{ik}}
                      =0\  \hspace{10mm}\numberthis\label{eq:zero_gradient_revenue}
                    \end{flalign*}
                    }
                    \endminipage
                    }
                  \end{center}
                  \caption{
                  {\em Gradient of $\mathrm{rev}_{\mathrm{shift},j}(\cdot)$.}
                  Equations from the proof of  Lemma~\ref{lem:rev_is_Lipschitz_in_shifts}.
                  }
                  \label{fig:missing_equations_from_rev_is_Lipschitz}
                \end{figure*}
                We can observe that every term in the gradient (Equation~\eqref{eq:Lem67eq1},\ Equation~\eqref{eq:zero_gradient_revenue}) is a linear function of $f_{ij}(\cdot)$ and $F_{ij}(\cdot)$ for some $i\in [n]$ and $j\in [m]$.
                Since, each term in the gradient (Equation~\eqref{eq:Lem67eq1}) involves at most $2n$ terms of the form of Equation~\eqref{eq:Lem67eq2} for some $i,k,\ell\in [n]$ and $j\in [m]$,
                \begin{align*}
                  \int\limits_{\mathrm{supp}(f_{ij})}\hspace{-4mm} y f_{ij}(y)f_{kj}(y+\alpha_{ij}-\alpha_{kj})
                  \hspace{-1mm}\prod _{\ell \neq i,k}F_{\ell j}(y+\alpha_{ij}-\alpha_{\ell j})dy.\numberthis\label{eq:Lem67eq2}
                \end{align*}
                Bounding this term, for all $i,k,\ell\in [n]$ and $j\in [m]$ by $\mu_{\max}\rho$ would give us a bound on $\nabla \mathrm{rev}_{\mathrm{shift}}(\alpha)$.
                To this end, consider
                \begin{align*}
                  \bigg|\hspace{-4mm}\int\limits_{\mathrm{supp}(f_{ij})}\hspace{-4mm}y f_{ij}(y)f_{kj}(y+\alpha_{ij}-\alpha_{kj})
                  \hspace{-1mm}\prod _{\ell \neq i,k}\hspace{-1mm}F_{\ell j}(y+\alpha_{ij}-\alpha_{\ell j})dy\bigg|
                  &\stackrel{(\hyperref[asmp:distributed_distribution]{18})}{\leq}
                  \mu_{\max}\bigg|\hspace{-4mm}\int\limits_{\mathrm{supp}(f_{ij})}\hspace{-4mm} y f_{ij}(y)dy\bigg| & ({\rm Using \ }F_{ij}(\cdot)\leq 1)\\
                  &\stackrel{(\hyperref[asmp:bounded_bid]{20})}{\leq} \mu_{\max}\rho.\numberthis\label{eq:Lem67eq4}
                \end{align*}
                For all $k\in[n]$, let terms $t_1(k)$ and $t_2(k)$ be defined as follows
                \begin{align*}
                  &t_1(k)\coloneqq\hspace{-4mm}\hspace{-1mm}\int\limits_{\mathrm{supp}(f_{ij})}
                  \hspace{-5mm}y f_{ij}(y)f_{kj}(y+\alpha_{ij}-\alpha_{kj})\prod _{\ell \neq i,k}F_{\ell j}(y+\alpha_{ij}-\alpha_{\ell j})dy\numberthis\label{eq:term1}\\
                  &t_2(k)\coloneqq\hspace{-6mm}\int\limits_{\mathrm{supp}(f_{kj})}\hspace{-5mm}y f_{kj}(y)f_{ij}(y+\alpha_{kj}-\alpha_{ij})
                  \prod _{\ell \neq i,k}F_{\ell j}(y+\alpha_{kj}-\alpha_{\ell j})dy.\numberthis\label{eq:term2}
                \end{align*}
                Then rewriting the gradient, from Figure~\ref{fig:missing_equations_from_rev_is_Lipschitz}, we have
                \begin{align*}
                  \bigg|\frac{\partial \mathrm{rev}_{\mathrm{shift},\ j}(\alpha)}{\partial \alpha _{ij}}\bigg| &= \hspace{1.5mm}\Pr_{\cU}[j]\sum_{k\in[n-1]\backslash \{i\} }\big(t_1(k)-t_2(k)\big)\\
                  &\hspace{-1mm}\stackrel{\eqref{eq:Lem67eq4}}{\leq} \Pr_{\cU}[j]\sum_{k\in [n-1]\backslash\{i\}}\mu_{\max}\rho\\
                  &\hspace{0.0mm}\stackrel{}{\leq} (n-2)\Pr_{\cU}[j]\ \rho\mu_{\max}\numberthis\label{eq:Lem67eq5}.
                \end{align*}
                Now calculating the Frobenius norm of $\mathrm{rev}_{\mathrm{shift,\ j}}(\alpha)$ we get
                \begin{align*}
                  &\hspace{-5.5mm}\|\nabla \mathrm{rev}_{\mathrm{shift},\ j}(\alpha)\|_F^2
                  =
                  \hspace{1mm}\sum_{\substack{i\in [n-1]\\ k\in [m]}}\bigg|\frac{\partial \mathrm{rev}_{\mathrm{shift},\ j}(\alpha)}{\partial \alpha _{ik}}\bigg|^2\\
                  &\hspace{23mm}\stackrel{\eqref{eq:zero_gradient_revenue}}{=}
                  \sum_{i\in [n-1]}\bigg|\frac{\partial \mathrm{rev}_{\mathrm{shift},\ j}(\alpha)}{\partial \alpha_{ij}}\bigg|^2\numberthis\label{eq:term3}\\
                  &\hspace{23mm}\stackrel{\eqref{eq:Lem67eq5}}{\leq} \Pr_{\cU}[j](n-1)((n-2)\rho\mu_{\max})^2.\numberthis\label{eq:gradrevjbound}
                \end{align*}
                Now, we proceed to bound $\nabla \mathrm{rev}_{\mathrm{shift}}(\alpha)$
                \begin{align*}
                  & \hspace{5mm} \|\nabla \mathrm{rev}_{\mathrm{shift}}(\alpha)\|_F^2
                  \stackrel{}{=}\hspace{1mm}
                    \sum_{\substack{i\in [n-1]\\ j\in [m]}}\bigg|\sum_{k\in [m]}\frac{\partial \mathrm{rev}_{\mathrm{shift},\ k}(\alpha)}{\partial \alpha _{ij}}\bigg|^2\\
                  &\hspace{30mm}\stackrel{\eqref{eq:zero_gradient_revenue}}{=}\sum_{\substack{i\in [n-1]\\ j\in[m]}}\bigg|\frac{\partial \mathrm{rev}_{\mathrm{shift},\ j}(\alpha)}{\partial \alpha_{ij}}\bigg|^2\\
                  &\hspace{30mm}\stackrel{\eqref{eq:gradrevjbound}}{=}\sum_{\substack{j\in[m]}}\|\mathrm{rev}_{\mathrm{shift},\ j}(\alpha)\|_F^2\\
                  &\hspace{31mm}\stackrel{}{\leq} (n-1)((n-2)\rho\mu_{\max})^2\sum_{j\in
                  [m]}\Pr_{\cU}[j]\\
                  &\hspace{31mm}\stackrel{}{\leq} (n-1)((n-2)\rho\mu_{\max})^2.\numberthis\label{eq:gradient_of_revenue_j_is_bounded}
                \end{align*}
                Therefore, it follows that $\|\nabla \mathrm{rev}_{\mathrm{shift}}(\alpha)\|_F \stackrel{}{\leq} n^{\frac{3}{2}}\rho\mu_{\max}$.
              \end{proof}

            \subsubsection{Proof of Lemma~\ref{lem:alpha_Lipschitz_in_coverage}}
              \label{app:proof_of_alpha_Lipschitz_in_coverage}

              \paragraph{Technical Lemmas.}
              We use the following lemmas in the proof of Lemma~\ref{lem:alpha_Lipschitz_in_coverage}.
              The first lemma is a lower bound on the derivative $q_{ij}(\alpha)$, and follows from assumptions~(\hyperref[asmp:eta_coverage]{17}) and (\hyperref[asmp:distributed_distribution]{18}).
              \begin{lemma}
                \label{lem:lower_bound_gradient_of_coverage}
                {\bf (Lower bound of derivative of $q_{ij}(\alpha)$).}
                If for all $i\in[n],\ j \in [m]$ the probability density function, $f_{ij}(\cdot)$, of virtual valuations is bounded below by $\mu_{\min}$,
                and every advertiser has at least $\eta$ coverage on every type $j\in [m]$,
                then the absolute value of each gradient $\big|\frac{\partial q_{ij}(\alpha)}{\partial \alpha_{sj}}\big|$ is lower bounded by $\eta\mu_{\min}$, i.e.,
                \begin{align*}
                  \bigg|\frac{\partial q_{ij}(\alpha)}{\partial \alpha_{sj}}\bigg| \geq \eta\mu_{\min} \ \forall \ i,s \in [n],\ j \in [m], \alpha \in \cR^{n\times m}.
                \end{align*}
              \end{lemma}
              \begin{proof}
              Each advertiser has at least $\eta$ coverage on every type, i.e., we have for all $i\in [n],\ j \in [m]$
              \begin{align*}
                &q_{ij}(\alpha) = \int\limits_{\mathrm{supp}(f_{ij})}\hspace{-5mm} f_{ij}(y)\prod_{k\in [n]\backslash \{i\}}  F_{kj}(y+\alpha_{ij}-\alpha_{kj})dy \stackrel{ }{\geq} \eta.
                \numberthis
                \label{eq:lower_bounding_coverage_2}
              \end{align*}
              Now considering $\frac{\partial q_{ij}(\alpha)}{\partial \alpha_{sj}}$ we get
              \begin{align*}
                &\hspace{26.5mm}\bigg|\frac{\partial q_{ij}(\alpha)}{\partial \alpha_{sj}}\bigg|=\bigg|\hspace{-3mm}\int\limits_{\mathrm{supp}(f_{ij})}\hspace{-5mm} f_{ij}(y)f_{sj}(y+\alpha_{sj}-\alpha_{ij})\prod_{k\neq
                i,s}  F_{kj}(y+\alpha_{ij}-\alpha_{kj})dy\hspace{1mm}\bigg|\\
                &\hspace{26.5mm}\hspace{15.5mm}\stackrel{}{\geq}\mu_{\min}\bigg|\hspace{-3mm}\int\limits_{\mathrm{supp}(f_{ij})}\hspace{-5mm} f_{ij}(y)\prod_{k\neq
                i,s}  F_{kj}(y+\alpha_{ij}-\alpha_{kj})dy\hspace{1mm}\bigg|\hspace{-26.5mm} \hspace{36.5mm} ({\rm Using \ } f_{ij}(\phi_{ij}) \geq \mu_{\min} )\\
                &\hspace{26.5mm}\hspace{15.5mm}\stackrel{}{\geq} \mu_{\min}\bigg|\hspace{-3mm}\int\limits_{\mathrm{supp}(f_{ij})}\hspace{-5mm} f_{ij}(y)\prod_{k\neq
                i}  F_{kj}(y+\alpha_{ij}-\alpha_{kj})dy\hspace{1mm}\bigg| \hspace{7mm}\hspace{-26.5mm}\hspace{36.5mm} ({\rm Using \ } F_{ij}(\phi_{ij}) \leq 1 ) \\
                &\hspace{26.5mm}\hspace{14.5mm}\stackrel{\eqref{eq:lower_bounding_coverage_2}}{\geq} \eta\mu_{\min}.\numberthis
              \end{align*}
            \end{proof}
              \noindent In the next lemma we extend the lower bound to the directional derivative of $q_{ij}(\alpha)$.

              \begin{lemma}\label{lem:unique_shift_for_shift}
                {\bf (Lower bound of directional derivative of $q_{ij}(\alpha)$).}
                Given a shift $\alpha_j \in \cR^{n-1}$, $t_{\max} > 0$, and  a direction vector $u \in \cR^{n-1}$ , s.t. %
                $\|u\|_2=1$,
                if the probability density function, $f_{ij}(\cdot)$, of virtual valuations is bounded below by $\mu_{\min}$ and bounded above by $\mu_{\max}$ $ \ \forall \ i\in[n],\ j \in [m]$,
                and $q_{ij}(tu+\alpha_j)>\eta$ for all$ \ t \in [0,t_{\max}]$,
                then for $i \in \argmax_{k\in [n-1]} |u_k|$ and for all $\ t \in [0,t_{\max}]$
                \begin{align}
                  \mathrm{sign}(u_{i})\frac{\partial q_{ij}(tu+\alpha_j))}{\partial t} > \frac{\eta\mu_{\min}}{\sqrt n}.
                \end{align}
              \end{lemma}
              \begin{proof}
                Consider $i \in \argmax_{k\in [n-1]} |u_k|$. Advertiser $i$'s bids are being increased faster than or equal to any other advertiser's.
                Recalling that the shift of advertiser $n$, $\alpha_{ij}=0$ for all user types $j\in[m]$, using Equation~(\hyperref[def:coverage2]{10}) we can express $ q_{ij}(tu+\alpha)$ and its gradient as shown in Figure~\ref{fig:missing_equations_from_unique_shift_for_shift}.
                \begin{figure*}
                  \centering
                  \fbox{
                  \minipage{0.95\linewidth}
                  \vspace{5mm}
                  \footnotesize
                  For all $i\in [n-1]$ and $j\in [m]$
                  \begin{align*}
                    &\hspace{3mm}q_{ij}(tu+\alpha) = \hspace{-1.5mm}\int\limits_{\mathrm{supp}(f_{ij})}\hspace{-5mm} f_{ij}(y)F_{nj}(y+tu_i+\alpha_{ij})\prod_{k\neq i,n}F_{kj}(y+t(u_i-u_k)+\alpha_{ij}-\alpha_{kj})dy\numberthis\\
                    &\frac{\partial q_{ij}(tu+\alpha)}{\partial t}
                    = u_i\hspace{-5mm}\int\limits_{\mathrm{supp}(f_{ij})}\hspace{-5mm} f_{ij}(y)f_{nj}(y+tu_i+\alpha_{ij})\hspace{0.8mm}\prod_{k\neq i,n}F_{kj}(y+t(u_i-u_k)+\alpha_{ij}-\alpha_{kj})dy\numberthis \label{eq:Lem611Eq1}\\
                    &\hspace{13mm}+(u_i\hspace{-0.5mm}-\hspace{-0.5mm}u_k)\hspace{-5mm}\int\limits_{\mathrm{supp}(f_{ij})}\hspace{-5mm} f_{ij}(y)F_{nj}(y\hspace{-0.5mm}+\hspace{-0.5mm}tu_i\hspace{-0.5mm}+\hspace{-0.5mm}%
                    \alpha_{ij})\hspace{-2mm}\sum_{k\neq i,n}\hspace{-1.5mm} f_{kj}(y\hspace{-0.5mm}+\hspace{-0.5mm}t(u_i\hspace{-0.5mm}-\hspace{-0.5mm}u_k)\hspace{-0.5mm}+\hspace{-0.5mm}\alpha_{ij}\hspace{-0.5mm}-\hspace{-0.5mm}\alpha_{kj})
                    \hspace{-2mm}\prod_{\ell\neq i,k,n}\hspace{-2mm}F_{\ell j}(y\hspace{-0.5mm}+\hspace{-0.5mm}t(u_i-u_\ell )\hspace{-0.5mm}+
                    \hspace{-0.5mm}\alpha_{ij}\hspace{-0.5mm}-\hspace{-0.5mm}\alpha_{\ell j})dy
                  \end{align*}
                  \endminipage
                  }
                  \caption{
                  {\em Directional Derivative of $q_{ij}(\cdot)$.}
                  Equations from the proof of Lemma~\ref{lem:unique_shift_for_shift}.}
                  \label{fig:missing_equations_from_unique_shift_for_shift}
                  \end{figure*}\\
                  Since $i \in \argmax_{k\in [n-1]} |u_k|$, we have
                  \begin{align*}
                    |u_i| &\geq |u_k|\\
                    \mathrm{sign}(u_i)u_i &\geq \max(u_k,-u_k)\\
                    (u_i - u_k)\mathrm{sign}(u_{i}) &> 0.\label{eq:Lem611Eq2}\numberthis
                  \end{align*}
                  Since $\|u\|_2=\sum_{i\in [n-1]}|u_i|^2= 1$, we can lower bound $|u_i|^2$, the maximum coordinate of $u\in\cR^{n-1}$ by magnitude by $\frac{1}{n-1}$, i.e., $ |u_i|\geq \nfrac{1}{\sqrt{n-1}}$.
                  Multiplying Equation~\eqref{eq:Lem611Eq1} with $\mathrm{sign}(u_{i})$ and using Equation~\eqref{eq:Lem611Eq2} and the fact that the integrals involved are positive to lower bound the equation we get
                  \begin{flalign*}
                    &\hspace{38.0mm}\mathrm{sign}(u_i) \frac{\partial q_{ij}(tu+\alpha_j)}{\partial t}
                    \stackrel{\eqref{eq:Lem611Eq2}}{\geq} \mathrm{sign}(u_i)u_i\frac{\partial q_{ij}}{\partial \alpha_{nj}}\bigg|_{\alpha_j+tu}&\\
                    &\hspace{74.5mm}\hspace{-5mm}\stackrel{\rm Lemma~\ref{lem:lower_bound_gradient_of_coverage}}{\geq} |u_i| \eta\mu_{\min}&\\
                    &\hspace{75mm}\stackrel{}{\geq} \frac{\eta\mu_{\min}}{\sqrt{n-1}}  \hspace{36.5mm}({\rm Using \ }|u_i| > \nfrac{1}{\sqrt{n-1}})\\
                    &\hspace{75mm}\stackrel{}{>} \frac{\eta\mu_{\min}}{\sqrt{n}}.&
                  \end{flalign*}
                  \vspace{-5mm}
                \end{proof}
              \begin{proof}[\unskip\nopunct]
                  \textit{Proof of Lemma~\ref{lem:alpha_Lipschitz_in_coverage}.}
                  Consider a type $j\in [m]$ and the corresponding shifts $\alpha_j,\beta_j\in \cR^n$, where $\alpha_j,\beta_j$ are the $j$-th columns of $\alpha$ and $\beta$ respectively.
                  Let $u\coloneqq\alpha_j-\beta_j$, then from Lemma~\ref{lem:unique_shift_for_shift} we have $\exists \ i \in [n-1]$ such that
                  \begin{align}
                    \label{eq:closealp}
                    \forall \ t \in [0,1], \ \bigg|\frac{\partial q_{ij}(tu+\beta_j)}{\partial t}\bigg|> \eta\mu_{\min}.
                  \end{align}
                  Consider this $i$, then from the fundamental theorem of calculus we have
                  \begin{align*}
                    &\hspace{5mm}\|q_{j}(\alpha_j)-q_{j}(\beta_j)\|_2^2
                    = \sum_{k\in [n]}\bigg|\int_{0}^{1}\frac{\partial q_{kj}(tu+\beta_j)}{\partial t}dt\bigg|^2\\
                    &\hspace{5mm}\hspace{30.5mm}\geq \bigg|\int_{0}^{1}\frac{\partial q_{ij}(tu+\beta_j)}{\partial t}dt\bigg|^2\\
                    &\hspace{5mm}\hspace{30.5mm}\hspace{-1mm}\stackrel{\eqref{eq:closealp}}{>} \frac{(\eta\mu_{\min})^2}{n}\bigg|\int_{0}^{1}(tu+\beta_j) dt\bigg|\\
                    &\hspace{5mm}\hspace{30.5mm}\stackrel{}{>} \frac{(\eta\mu_{\min})^2}{n} \|\alpha_j-\beta_j\|_2^2.\numberthis\label{eq:coverage_of_j_is_Lipschitz}
                  \end{align*}
                  Using Equation~\eqref{eq:coverage_of_j_is_Lipschitz} for every type $j\in[m]$ we get that $\|q(\alpha)-q(\beta)\|_F > \nfrac{(\eta\mu_{\min})^2}{n}\cdot\|\alpha-\beta\|_F$.
                \end{proof}

          \subsection{Proof of Lemma~\ref{thm:3}}
            \label{sec:algorithm2}
            \begin{algorithm}[t]
              \caption{ Algorithm2($\delta,\alpha_1,\xi, L,\eta,\mu_{\max},\mu_{\min}$)}
              \label{Algorithm2}
              \textbf{Input:}
                A target coverage $\delta\in[0,1]^{n\times m}$,
                an approximate shift $\alpha_1\in\cR^{n\times m}$,
                a constant $\xi>0$ that controls the accuracy,
                %
                %
                Lipschitz constant $L>0$ of $f_{ij}(\cdot)$,
                the minimum coverage $\eta>0$,
                and the lower and upper bounds, $\mu_{\min}$ and $\mu_{\max}$, of $f_{ij}(\cdot)$.\\
              \textbf{Output:} An approximation $\alpha\in\cR^{n\times m}$ of shifts for $\delta$.

              \begin{algorithmic}[1]
                \STATE Initialize $\gamma\coloneqq (4nL+2n^3\mu_{\max }^2)^{-1}$
                \STATE Initialize \\ \quad$T\coloneqq \log\big({\nfrac{mn^3\cL(\alpha_1)}{\epsilon}}\big)\frac{(L+n^2\mu_{\max}^2)}{(\eta\mu_{\min})^2}$

                \FOR{t = 1,2,\dots,T}
                \STATE Compute $\nabla L(\alpha_{t})\coloneqq\nabla [\mathcal{L}_1(\alpha_{t}),\dots,\mathcal{L}_m(\alpha_{t})]$
                \STATE  Update $\alpha_{t+1}\coloneqq\ \alpha_{t}-\gamma\nabla L(\alpha_{t})$
                \ENDFOR
                \STATE \bf return $\alpha$
              \end{algorithmic}
            \end{algorithm}
               \noindent The following remark gives insight, allowing us to use a gradient-based algorithm to solve the {\em optimal shift problem}.
               \begin{remark}
                 \label{rem:properties_of_non_convex_loss}
                 We observe that $\nabla \cL=2\sum_{i,j}(q_{ij}(\alpha)-\delta_{ij})\nabla q_{ij}(\alpha)$ is a linear combination of the rows, $\{\nabla q_{ij}\}$, of $J_{q}(\alpha)$.
                 Since $\{\nabla q_{ij}(\alpha)\}$ are linearly independent, $\nabla \cL\neq0$ unless we are at the global minimum where $\delta=q(\alpha)$.
                 This guarantees that $\cL(\cdot)$ does not have any saddle-points or local-maxima, and that any local minimum is a global minimum.
               \end{remark}
             \noindent Now, to get an efficient complexity with a gradient-based algorithm we want to avoid small gradients ``far" from the optimal.
             Lemma~\ref{lem:lower_bounding_gradient_of_loss} shows that if $\cL(\alpha)$ greater than $\epsilon$, then the Frobenius norm $\|\cL(\alpha)\|_F$ of $\cL(\alpha)$ is greater than $\sqrt \epsilon$.
             The proof of Lemma~\ref{lem:lower_bounding_gradient_of_loss} is provided in Section~\ref{app:proof_of_lower_bounding_gradient_of_loss}.
            The proof of Lemma~\ref{lem:lower_bounding_gradient_of_loss} is provided in Section~\ref{app:proof_of_lower_bounding_gradient_of_loss}.
            \begin{lemma}
              \textbf{(Lower bounding $\nabla \cL_j(\cdot)$).}\label{lem:lower_bounding_gradient_of_loss}
              Given $\alpha_j\in \cR^{n-1}$,
              such that $\cL_j(\alpha_j) > \epsilon$ and $q_{ij}(\alpha_j) > \eta$,
              if the probability density function, $f_{ij}(\cdot)$, of virtual valuations is bounded below by $\mu_{\min}$ $ \ \forall \ i\in[n],\ j \in [m]$,
              then $\|\nabla \cL_j(\alpha_j)\|_2 > \frac{2}{n-1}\sqrt \epsilon \eta\mu_{\min}.$
            \end{lemma}
            Next, in Lemma~\ref{lem:revenue_is_Lipschitz_in_shift} we show that the gradient, $\nabla \mathcal{L}(\alpha)$, is $O(n(L+n^2\mu_{\max}^2))$-Lipschitz continuous.
            Therefore, at each step where $\cL(\alpha)\geq \xi$, we improve the loss by a factor of  $1-\beta\xi$, where $\beta$ does not depend on $\xi$.
            This gives us a complexity bound of $O(\log\nfrac{1}{\epsilon})$.
            The proof of Lemma~\ref{lem:revenue_is_Lipschitz_in_shift} is presented in Section~\ref{app:proof_of_revenue_is_Lipschitz_in_shift}.
            \begin{lemma}
              \label{lem:revenue_is_Lipschitz_in_shift}
              \textbf{(Gradient of $\cL(\cdot)$ is Lipschitz).}
              If the probability density function, $f_{ij}(\phi)$, of the virtual valuations, $\phi_{ij}$ is $L$-Lipschitz continuous and bounded above by $\mu_{\max}$,
              then
              $\nabla\cL_j(\alpha_j)$ is $O(n(L+n^2\mu_{\max}^2))$-Lipschitz.
            \end{lemma}
            \noindent Now, at each step, if the loss is greater than $\xi$, we get an improvement by a factor of  $1-\beta\xi$, where $\beta$ does not depend on $\xi$.
            This gives us a complexity bound of $O(\log\nfrac{1}{\epsilon})$.
            \label{app:proof_of_thm:3}
            \begin{proof}[Proof of Lemma~\ref{thm:3}]
              At each iteration of the algorithm we calculate $\nabla \cL_j(\alpha)$ for all $j\in [m]$, i.e., we calculate $\nabla \cL(\alpha)$.
               We note that this bounds the arithmetic calculations at one iteration.
               We recall from Equation~\eqref{eq:Jacobian_is_sparse} that the shift for one user type do not affect the coverage for the other.
               Therefore we can independently find a optimal shift $\alpha_j$ for all each user type $j\in [m]$.\\
               From Lemma~\ref{lem:revenue_is_Lipschitz} we have that $\cL_j$ is $O(n(L+n^2\mu_{\max}^2))$-Lipschitz continuous.
               Let $L^\prime \coloneqq O(n(L+n^2\mu_{\max}^2))$, for brevity.
               We can get an upper bound to $\cL_j(\alpha_k)$ from the first order approximation of $\cL_j$ at $\alpha_k$, further using the update rule $\alpha_{k+1}=\alpha_{k}-\frac{1}{L^\prime}\nabla \cL_j(\alpha_k)$ we have
               \begin{align*}
                 &\hspace{8mm}\cL_j(\alpha_{k+1})\leq \cL_j(\alpha_k)-\frac{1}{2L^\prime} \|\nabla \cL_j(\alpha_k)\|_2^2.
               \end{align*}
               Let $\lambda\coloneqq \frac{2}{n-1}\eta\mu_{\min}$,
               then from Lemma~\ref{lem:lower_bounding_gradient_of_loss} we have that $\nabla \cL_j(\alpha)$ is lower bounded by $\sqrt{\cL_j(\alpha_{k})} \lambda.$
               Using this to lower bound the gradient we get
               \begin{align*}
                 &\hspace{13mm}\hspace{-18mm}\cL_j(\alpha_k)- \cL_j(\alpha_{k+1})\geq \frac{1}{2L^\prime} \|\nabla \cL_j(\alpha_k)\|_2^2\\
                 &\hspace{13mm}\hspace{-3mm}\cL_j(\alpha_{k+1})\leq \cL_j(\alpha_k)-\frac{\cL_j(\alpha_{k}) \lambda^2}{2L^\prime}.
               \end{align*}
               By the above recurrence we get
               \begin{align*}
                 \cL_j(\alpha_k) \leq \cL_j(\alpha_0)\bigg(1-\frac{\lambda^2}{2L^\prime}\bigg)^k.
               \end{align*}
               Setting $k\coloneqq \log {\frac{m\cL(\alpha_0)}{\epsilon}}\frac{-1}{\log{\big(1-\frac{\lambda^2}{2L^\prime}\big)}}$ we get that for all $j\in [m],$ $\cL_j(\alpha_k)<\nfrac{\epsilon}{m}$. Therefore
               \begin{align*}
                 &\hspace{-1mm}\cL(\alpha)=\sum_{j=1}^{m}\cL_j(\alpha_j)
                   \stackrel{}{<}\epsilon.
               \end{align*}
               Substituting $L^\prime = O(n(L+n^2\mu_{\max}^2))$ we get that
               the algorithm outputs $\alpha$,
               such that $\cL(\alpha)<\epsilon$
               in
               \begin{align*}
                 \log\bigg({\frac{m\cL(\alpha_1)}{\epsilon}}\bigg)\frac{n^3(L+n^2\mu_{\max}^2)}{(\eta\mu_{\min})^2}\ {\rm steps.}
               \end{align*}
              \end{proof}
              \subsubsection{Proof of Lemma~\ref{lem:lower_bounding_gradient_of_loss}}
              \label{app:proof_of_lower_bounding_gradient_of_loss}
              In the proof of Lemma~\ref{lem:lower_bounding_gradient_of_loss} we use Lemma~\ref{lem:lower_bounding_linear_combination_2}, which shows that any linear combination of $\nabla q_{ij}(\alpha)$ for all $i\in[n]$, with reasonably ``large" weights is lower bounded.
              We note that Lemma~\ref{lem:lower_bounding_linear_combination_2} does not follow from linear independence of $\nabla q_{ij}(\alpha) \ \forall \ i\in[n]$ (Lemma~\ref{lem:Jacobian_is_invertible}), because linear combinations of linearly independent vectors can be arbitrary small while having ``large" weights.
              \begin{lemma}
                \label{lem:lower_bounding_linear_combination_2}
                Given $x\in\cR^{n-1}$ such that $\|x\|_1>1$,
                if for all $i\in[n],\ j \in [m]$ the probability density function, $f_{ij}(\cdot)$, of virtual valuations is bounded below by $\mu_{\min}$,
                and $q_{ij}(\alpha_j)>\eta$ coverage on every user type $j\in [m]$,
                then
                \begin{align}
                  \bigg\|\sum_{i\in [n-1]}x_i\nabla q_{ij}(\alpha_j)\ \bigg\|_2>\frac{\eta\mu_{\min}}{n-1} \ \forall \ \alpha_j \in \cR^{n-1}.
                \end{align}
              \end{lemma}
              \begin{proof}
                Without loss of generality consider a reordering of $(x_1,x_2,\dots,x_n)$, s.t. for some $p\leq n-1$
                \begin{align}
                  &x_i \geq 0 \ \forall \ i\leq p\label{eq:lem11eq1}\\
                  &\hspace{0.3mm}x_i < 0 \ \forall \ i > p.\label{eq:lem11eq2}
                \end{align}

                \noindent{\bf Case A. $\sum_{i\in[p]}x_i < \nfrac{1}{2}$ :}\\
                We can replace $x$ by $-x$, since this does not change the norm $\big\|\sum_{i\in [n-1]}x_i\nabla q_{ij}(\alpha_j)\big\|_2$.
                Now replacing  $p$ by $(n-p-1)$ we get {\em case B}.

                \vspace{2mm}

                \noindent{\bf Case B. $\sum_{i\in[p]}x_i \geq \nfrac{1}{2}$ :}\\
                The coverage remains invariant if the bids of all advertisers are uniformly shifted for any given user type $j$.
                $(\alpha_{1j},\alpha_{2j},\dots,\alpha_{nj})$. Therefore we have for all $i\in[n-1]$
                \begin{align*}
                  &\hspace{-20mm}\frac{\partial q_{ij}(\alpha)}{\partial \alpha_{ij}} + \hspace{-3mm} \sum_{k \in [n-1]\backslash \{i\}} \hspace{-2mm}  \frac{\partial q_{ij}(\alpha)}{\partial \alpha_{kj}}\stackrel{\eqref{eq:coverage_as_probability_1}}{=} - \frac{\partial q_{ij}(\alpha)}{\partial \alpha_{n j}}\numberthis\\
                  &\hspace{-20mm}\hspace{42.5mm}\stackrel{\eqref{eq:negative_gradient}}{=}\bigg|\frac{\partial q_{ij}(\alpha)}{\partial \alpha_{n j}}\bigg|\\
                  &\hspace{-20mm}\hspace{42.5mm}\hspace{-5mm}\stackrel{{\rm Lemma}~\ref{lem:lower_bound_gradient_of_coverage}}{\geq}\eta\mu_{\max}. \label{eq:rowSum}
                \end{align*}
                Calculating the weighted sum of Equation~\eqref{eq:rowSum} over $i \in [p]$ with weights $x_i$ we get
                \begin{align*}
                  &\hspace{-16mm}\sum_{i\in [p]}x_i\bigg(\sum_{k\in [n-1]}\frac{\partial q_{ij}(\alpha)}{\partial \alpha_{kj}}\bigg)
                  \stackrel{\eqref{eq:rowSum}}{>} \sum_{i\in[p]}x_i\eta\mu_{\min}\\
                  &\hspace{-16mm}\hspace{43mm}> \frac{\eta\mu_{\min}}{2}.
                \end{align*}
                On rearranging the LHS we get
                \begin{align*}
                  &\hspace{-23.5mm}\sum_{k\in [n-1]}\bigg(\sum_{i\in [p]}x_i\frac{\partial q_{ij}(\alpha)}{\partial \alpha_{kj}}\bigg) > \frac{\eta\mu_{\min}}{2}.
                \end{align*}
                Therefore, by the pigeonhole principle on elements of the outer sum, $\exists \ k \in [n-1],$ s.t.
                \begin{align}
                  \sum_{i\in[p]}x_i\frac{\partial q_{ij}(\alpha)}{\partial \alpha_{kj}}&\geq \frac{1}{n-1}\sum_{k\in [n-1]}\hspace{-2mm}\bigg(\sum_{i\in [p]}x_i\frac{\partial q_{ij}(\alpha)}{\partial \alpha_{kj}}\bigg)\\
                  &\geq  \frac{\eta\mu_{\min}}{2(n-1)}.
                  \label{eq:tmp}
                \end{align}
                \noindent From Equation~\eqref{eq:negative_gradient} for all $i \in [p]$ and $k > p$, $\frac{\partial q_{ij}(\alpha)}{\partial \alpha_{kj}}<0$ .
                Therefore, $k\leq p$ in Equation~\eqref{eq:tmp}.
                From this we get
                \begin{align*}
                  &\sum_{i\in[n-1]}x_i\frac{\partial q_{ij}(\alpha)}{\partial \alpha_{kj}} = \sum_{i\in[p]}x_i\frac{\partial q_{ij}(\alpha)}{\partial \alpha_{kj}}+\sum_{i\in[n-1]\backslash[p]}x_i\frac{\partial q_{ij}(\alpha)}{\partial \alpha_{kj}}\\
                  &\hspace{26.5mm}\hspace{-1mm}\stackrel{\eqref{eq:tmp}}{\geq}\hspace{5.5mm} \frac{\eta\mu_{\min}}{2(n-1)} +\sum_{i\in[n-1]\backslash[p]}x_i\frac{\partial q_{ij}(\alpha)}{\partial \alpha_{kj}}\\
                  &\hspace{26.5mm}\hspace{-5mm}\stackrel{{\rm Lemma}~\ref{lem:lower_bound_gradient_of_coverage}}{\geq} \frac{\eta\mu_{\min}}{2(n-1)} +\eta\mu_{\min}\hspace{-3mm}\sum_{i\in [n-1]\backslash [p]}\hspace{-3mm}(-x_i)\\
                  &\hspace{26.5mm}\hspace{-1mm}\stackrel{\eqref{eq:lem11eq2}}{\geq} \hspace{5.5mm}\frac{\eta\mu_{\min}}{2(n-1)}.\numberthis\label{eq:lem11eq3}
                \end{align*}
                Therefore, $\exists \ k \in [n-1]$ such that $\sum_{i\in[n-1]}x_i\frac{\partial q_{ij}(\alpha)}{\partial \alpha_{kj}}>\eta\mu_{\min}$.
                It follows that
                \begin{align*}
                  &\hspace{-4mm}\bigg\|\sum_{i\in[n-1]}\hspace{-3mm}x_i\nabla q_{ij}(\alpha)\hspace{1mm}\bigg\|_2^2 = \sum_{t\in [m]}\sum_{k\in [n-1]}\bigg(\sum_{i\in[n-1]} x_i\frac{\partial q_{ij}(\alpha)}{\partial \alpha_{kt}}\bigg)^2\\
                  &\hspace{-4mm}\hspace{31.5mm}\hspace{-0.5mm}\stackrel{\eqref{eq:Jacobian_is_sparse}}{=} \sum_{k\in [n-1]}\hspace{-2mm}\bigg(\sum_{i\in[n-1]} x_i\frac{\partial q_{ij}(\alpha)}{\partial \alpha_{kj}}\bigg)^2\\
                  &\hspace{-4mm}\hspace{31.5mm}\hspace{-0.5mm}\stackrel{\eqref{eq:lem11eq3}}{\geq} \bigg(\frac{\eta\mu_{\min}}{2(n-1)}\bigg)^2\numberthis\label{eq:result_lemA6}\\
                  &\hspace{-4mm}\hspace{-2.5mm}\bigg\|\sum_{i\in[n-1]}x_i\nabla q_{ij}(\alpha)\hspace{1mm}\bigg\|_2 \geq \frac{\eta\mu_{\min}}{2(n-1)}.
                \end{align*}
              \end{proof}

              \begin{proof}[\nopunct\unskip]
                {\bf Proof of Lemma~\ref{lem:lower_bounding_gradient_of_loss}.}
                Since $\cL_j(\alpha_j)\geq \epsilon$, we have
                \begin{align}
                  \label{eq:lower_bounding_loss}
                  &\hspace{5mm}\cL_j(\alpha_j) = \sum_{i\in [n-1]}(\delta_{ij}-q_{ij}(\alpha_j))^2 \geq  \epsilon.
                \end{align}
                Further, using $\big(\sum_{i}a_i \big)^2 =\sum_{i}a_i^2 + \sum_{\substack{i,k}}2a_ia_k$ we get
                \begin{align*}
                  &\hspace{5.5mm}\cL_j(\alpha_j) = \sum_{i\in [n-1]}(\delta_{ij}-q_{ij}(\alpha_j))^2\\
                  &\hspace{5.5mm}\hspace{11mm}\leq \bigg(\sum_{i\in [n-1]}\big|\delta_{ij}-q_{ij}(\alpha_j)\big|\bigg)^2.\numberthis \label{eq:lower_bounding_loss2}
                \end{align*}
                From these we have that
                \begin{align}
                  &\hspace{-54mm}\sum_{i\in [n-1]}|\delta_{ij}-q_{ij}(\alpha_j)| \stackrel{\eqref{eq:lower_bounding_loss2},\eqref{eq:lower_bounding_loss}}{\geq} \sqrt\epsilon.
                \end{align}
                Considering $x_i=\frac{1}{\sqrt\epsilon}(\delta_{ij}-q_{ij}(\alpha_j))$ we have
                \begin{align*}
                  &\hspace{3mm}\sum_{i\in[n-1]}|x_i|=\frac{1}{\sqrt\epsilon}\sum_{i\in[n-1]}|\delta_{ij}-q_{ij}(\alpha)|>1.
                \end{align*}
                From Lemma~\ref{lem:lower_bounding_linear_combination_2} we have
                \begin{align*}
                  &\hspace{21mm}\bigg\|\sum_{i\in[n-1]}x_i\nabla q_{ij}(\alpha_j)\hspace{1.5mm}\bigg\|_2 \stackrel{\rm Lemma~\ref{lem:lower_bounding_linear_combination_2}}{\geq} \frac{\eta\mu_{\min}}{n-1}\\
                  &\bigg\|\sum_{i\in[n-1]}2(\delta_{ij}-q_{ij}(\alpha_j))\nabla q_{ij}(\alpha_j)\hspace{1.5mm}\bigg\|_2 \hspace{6mm}\geq 2\sqrt\epsilon\frac{\eta\mu_{\min}}{n-1}.
                \end{align*}
              \end{proof}

              \subsubsection{Proof of Lemma~\ref{lem:revenue_is_Lipschitz_in_shift}}
                \label{app:proof_of_revenue_is_Lipschitz_in_shift}
                In order to show that the loss $\cL(\cdot)$ is $O(n(L+n^2\mu_{\max}^2))$-Lipschitz continuous, we first show that $\nabla q_{ij}$ is $2n(L+n\mu_{\max}^2)$-Lipschitz continuous.
                To this end, we show that the elements of $\nabla^2 q_{ij}$ are bounded (Lemma~\ref{lem:bounded_element}), and then use Lemma~\ref{lem:smallLips} (Corollary~{1.2} in \cite{varga}) to bound the magnitudes of the eigen-values.
                \begin{lemma}
                  \label{lem:bounded_element}
                  Given $\alpha_j \in \cR^{n}$, if $\mathrm{pdf}$, $f_{ij}(\phi)$ of the virtual valuations, $\phi_{ij}$ is $L$-Lipschitz continuous and bounded above by $\mu_{\max}$,
                  then elements in the main diagonal of the Hessian, $\nabla^2 q_{ij}(\alpha_j)$ are bounded in absolute value by $n(L+n\mu_{\max }^2)$, and all other elements are bounded in absolute value by $L+n\mu_{\max}^2$, i.e.,
                  \begin{align*}
                    &\forall\ i \in [n], \quad\quad\quad\quad\quad\quad\quad\quad\ \  \frac{\partial^2 q_{ij}}{\partial \alpha_{ij}\partial \alpha_{ij}} \leq n(L+n\mu_{\max }^2)\\
                    &\forall\ k,\ t  \in [n], k\neq i \mathrm{\ or\ }t \neq i,\quad \frac{\partial^2 q_{ij}}{\partial \alpha_{kj}\partial \alpha_{tj}} \leq L+n\mu_{\max }^2.
                  \end{align*}
                \end{lemma}
                \begin{proof}
                  Consider the Hessian of $q_{ij}(\alpha_j)$ in Figure~\ref{fig:missing_equations_from_bounded_element}, which follows from differentiating Equation~\hyperref[def:coverage2]{10} with respect to $\alpha_j$, where $\alpha_j$ is the $j$-th column of $\alpha$.
                  We note that $\frac{q_{ij}}{\alpha_{st}}=0$ for any $t\neq j$, for all $i,s\in [n]$ and $j,t\in[m],$ and so we only need to calculate the gradient with respect to $\alpha_j$.
                  \begin{figure*}
                    \begin{center}
                      \fbox{
                      \minipage{0.95\linewidth}
                      \footnotesize
                      \vspace{5mm}
                      \ \ For all distinct $i,k,t$ in $[n]$
                      \begin{align*}
                        \label{eq:diagonal_term_Hessian}
                        \frac{\partial^2 q_{ij}}{\partial \alpha_{ij}\partial \alpha_{ij}}&=\hspace{-1mm}\int\limits_{\mathrm{supp}(f_{ij})}\hspace{-5mm} f_{ij}(y)\sum_{k\neq
                        i}f_{kj}^\prime  (y+\alpha_{ij}-\alpha_{kj})\prod_{\ell\neq k,i}F_{\ell j}(y+\alpha_{ij}-\alpha_{\ell j})dy \numberthis \\
                        & \quad +\hspace{-4mm} \int\limits_{\mathrm{supp}(f_{ij})}\hspace{-5mm} f_{ij}(y)\sum_{k\neq
                        i}f_{kj}(y+\alpha_{ij}-\alpha_{kj})\sum_{\ell\neq i,k} f_{\ell j}(y+\alpha_{ij}-\alpha_{\ell j})\prod_{h \neq \ell,k,i} F_{hj}(y+\alpha_{ij}-\alpha_{hj})dy\\
                        \frac{\partial^2 q_{ij}}{\partial \alpha_{kj}\partial \alpha_{kj}}&= \hspace{-1mm}\int\limits_{\mathrm{supp}(f_{ij})}\hspace{-5mm} f_{ij}(y)_{
                        }f_{kj}^\prime(y+\alpha_{ij}-\alpha_{kj})\prod_{\ell\neq k,i}F_{\ell j}(y+\alpha_{ij}-\alpha_{\ell j})dy \numberthis\\
                        \frac{\partial^2 q_{ij}}{\partial\alpha_{kj}\partial \alpha_{ij}} = \frac{\partial^2 q_{ij}}{\partial \alpha_{ij}\partial\alpha_{kj}} &=-\hspace{-4mm}\int\limits_{\mathrm{supp}(f_{ij})}\hspace{-5mm} f_{ij}(y)_{
                        }f_{kj}^\prime(y+\alpha_{ij}-\alpha_{kj})\prod_{\ell\neq k,i}F_{\ell j}(y+\alpha_{ij}-\alpha_{\ell j})dy \numberthis \\
                        &\quad-\hspace{-4mm}\int\limits_{\mathrm{supp}(f_{ij})}\hspace{-5mm} f_{ij}(y)_{
                        }f_{kj}(y+\alpha_{ij}-\alpha_{kj}) \sum_{\ell\neq k,i} f_{\ell j}(y+\alpha_{ij}-\alpha_{\ell j})\prod_{h \neq \ell,k,i}F_{hj}(y+\alpha_{ij}-\alpha_{hj})dy\\
                        \frac{\partial^2 q_{ij}}{\partial \alpha_{kj}\partial \alpha_{tj}}&=\hspace{-1mm} \int\limits_{\mathrm{supp}(f_{ij})}\hspace{-5mm} f_{ij}(y)_{
                        }f_{kj}(y+\alpha_{ij}-\alpha_{kj})f_{tj}(y+\alpha_{ij}-\alpha_{tj})\prod_{\ell\neq k,i}F_{\ell j}(y+\alpha_{ij}-\alpha_{\ell j})dy. \numberthis
                      \end{align*}
                      \endminipage
                      }
                    \end{center}
                    \caption{
                    {\em Hessian of $q_{ij}(\cdot)$.}
                    Equations from proof of Lemma~\ref{lem:bounded_element}.}
                    \label{fig:missing_equations_from_bounded_element}
                  \end{figure*}

                  We can observe that for all $i \in [n]$ and $ j\in[m]$ every term in the Hessian is linear function of $f_{ij}^\prime(y)$, $f_{ij}(y)$ and $F_{ij}(y)$.
                  In particular each term in the Hessian is a sum of the following terms, for some combinations of $i,k,\ell\in [n]$ and $j\in [m]$
                  \begin{align*}
                    &\hspace{-5mm}\int\limits_{\mathrm{supp}(f_{ij})}\hspace{-5mm}f_{ij}(y)f_{kj}(y+\alpha_{ij}-\alpha_{kj})f_{\ell j}(y+\alpha_{ij}-\alpha_{\ell j})\prod_{h \neq \ell,k,i}\hspace{-1mm}F_{hj}(y+\alpha_{ij}-\alpha_{hj})dy
                    \numberthis\label{eq:term1_Hessian}\\
                    &\hspace{-5mm}\int\limits_{\mathrm{supp}(f_{ij})}\hspace{-5mm}f_{ij}(y)f_{kj}^\prime(y+\alpha_{ij}-\alpha_{kj})\prod_{\ell\neq k,i}F_{\ell j}(y+\alpha_{ij}-\alpha_{\ell j})dy.
                    \numberthis\label{eq:term2_Hessian}
                  \end{align*}
                  \noindent Each term along the diagonal of the Hessian (Equation~\eqref{eq:diagonal_term_Hessian}), $\frac{\partial^2 q_{ij}}{\partial \alpha_{ij}\partial \alpha_{ij}}$, is a combination of $(n-1)$ terms of the form Equation~\eqref{eq:term1_Hessian}, and $n^2$ terms of the form Equation~\eqref{eq:term2_Hessian}.
                  All other terms in the Hessian contain at most $n$ terms of the form Equation~\eqref{eq:term2_Hessian}, and 1 term of the form Equation~\eqref{eq:term1_Hessian}.
                  Bounding these terms for all $i,k,\ell\in [n]$ and $j\in [m]$ by $\mu_{\max}^2$  would give us a bound on terms of the Hessian, which in turn gives bounds on the eigen-values of the Hessain.

                  \noindent To this end, recall that for all $i\in[n], \ j \in [m],$ and $y\in \mathrm{supp}(f_{ij})$
                  \begin{align}
                    0 < f_{ij}(y) &\leq \mu_{\max},\label{eq:pdf_is_bounded2}\\
                    0 \leq F_{ij}(y) &\leq 1,\label{eq:cdf_is_bounded2}\\
                    |f_{ij}^\prime(y)| &< L,\label{eq:LemA8eq3}\\
                    \int\limits_{\mathrm{supp}(f_{ij})}\hspace{-5mm} f_{ij}(z) dz &= 1.\label{eq:LemA8Eq1}
                  \end{align}
                  We can now bound Equation~\eqref{eq:term1_Hessian} and Equation~\eqref{eq:term2_Hessian} as follows
                  \begin{align*}
                    \eqref{eq:term1_Hessian} &\stackrel{\eqref{eq:pdf_is_bounded2},\eqref{eq:cdf_is_bounded2}}{\leq}
                    \mu_{\max}^2\bigg|\int\limits_{\mathrm{supp}(f_{ij})}\hspace{-5mm} f_{ij}(y)dy\hspace{2mm}\bigg|
                    \stackrel{\eqref{eq:LemA8Eq1}}{\leq}\mu_{\max}^2\numberthis \label{LemA8eq5},\\
                    \eqref{eq:term2_Hessian}&\stackrel{\eqref{eq:LemA8eq3}, \eqref{eq:cdf_is_bounded2}}{\leq}
                    \hspace{2.5mm} L \hspace{3mm}\bigg|\int\limits_{\mathrm{supp}(f_{ij})}\hspace{-5mm} f_{ij}(y)dy\hspace{2mm}\bigg|
                    \stackrel{\eqref{eq:LemA8Eq1}}{\leq}L.\numberthis \label{LemA8eq4}
                  \end{align*}
                  Now we have for all $k,\ i\in[n]$, $\st,\ k\neq i$
                  \begin{align*}
                    &\hspace{34.4mm}\bigg|\frac{\partial^2 q_{ij}}{\partial \alpha_{ij}\partial \alpha_{ij}}\bigg|\hspace{5.5mm}=\hspace{4.5mm}\bigg|\sum_{k\neq
                    i}\eqref{eq:term2_Hessian}+\sum_{k\neq i}\sum_{\ell\neq i,k}\eqref{eq:term1_Hessian}\bigg|\\
                    &\hspace{34.4mm}\hspace{14.5mm}\hspace{7.8mm}\hspace{-4mm}\hspace{1.0mm}\stackrel{\eqref{LemA8eq5},\eqref{LemA8eq4}}{\leq} (n-1)\big(L+(n-2)\mu_{\max}^2\big)&\hspace{-3.5mm}(\mathrm{Using \ triangle \ inequality })\\
                    &\hspace{34.4mm}\hspace{14.5mm}\hspace{7.8mm}\hspace{1.0mm}\leq\hspace{4.5mm} n(L+n\mu_{\max}^2)\numberthis\\
                    &\hspace{34.4mm}\bigg|\frac{\partial^2 q_{ij}}{\partial \alpha_{kj}\partial \alpha_{kj}}\bigg|\hspace{4.5mm}=\hspace{4.5mm}\big| \eqref{eq:term2_Hessian}\big|\numberthis\hspace{7mm}\stackrel{\eqref{LemA8eq4}}{\leq}L\\
                    &\hspace{34.4mm}\bigg|\frac{\partial^2 q_{ij}}{\partial \alpha_{kj}\partial \alpha_{tj}}\bigg|\hspace{4.8mm}=\hspace{5mm} \big|\eqref{eq:term1_Hessian}\big|\hspace{7mm}\stackrel{\eqref{LemA8eq5}}{\leq} \mu_{\max}^2 \numberthis\\
                    &\hspace{34.4mm}\bigg|\frac{\partial^2 q_{ij}}{\partial\alpha_{kj}\partial \alpha_{ij}}\bigg|\hspace{4.8mm} =\hspace{5.0mm} \bigg|\frac{\partial^2 q_{ij}}{\partial \alpha_{ij}\partial\alpha_{kj}}\bigg|
                    =\bigg|\eqref{eq:term2_Hessian}+\sum_{\ell\neq k,i}\eqref{eq:term1_Hessian}\bigg|&\text{\white{.}}\\
                    &\hspace{34.4mm}\qquad\qquad\hspace{4mm}\stackrel{\eqref{LemA8eq5},\eqref{LemA8eq4}}{\leq} L+(n-2)\mu_{\max}^2\numberthis
                  \end{align*}
                \end{proof}
                \begin{lemma}
                  \label{lem:smallLips}
                  (Corollary~{1.2} in \cite{varga}) For any matrix $A\in\cR^{n\times n}$, and any eigen-value $\lambda\in \cR$ of $A$,
                  \begin{align*}
                    \lambda\leq \max_{i\in[n]}\sum_{j\in[n]}|A_{ij}|.
                  \end{align*}
                \end{lemma}
                \noindent We refer the reader to \cite{varga} for a proof of the above lemma.
                \begin{proof}[\unskip\nopunct]
                  \textbf{Proof of Lemma~\ref{lem:revenue_is_Lipschitz_in_shift}.}
                  To show that $\nabla \cL_j(\alpha_j)$ is Lipschitz continuous, we show that $q_{ij}(\alpha_{j})$ is Lipschitz continuous, then use the fact that $\nabla q_{ij}(\alpha_{j})$ is Lipschitz continuous from Lemma~\ref{lem:smallLips}, and that $\delta_j$ and $q_{ij}(\cdot)$ have bounded sums if the loss is greater than $\epsilon$.
                  To this end we recall
                  \begin{align*}
                    \cL_j(\alpha_j)&\coloneqq \sum_{i\in [n-1]}(\delta_{ij}-q_{ij}(\alpha_{j}))^2\\
                    \nabla \cL_j(\alpha_j) &= -2\sum_{i\in [n-1]}(\delta_{ij}-q_{ij}(\alpha_{j}))\nabla q_{ij}(\alpha_{j})
                  \end{align*}
                  Consider the following term for some $i,k\in[n]$ and $j\in [m]$
                  \begin{align*}
                    &t(k)\coloneqq\int\limits_{\mathrm{supp}(f_{ij})}\hspace{-5mm} f_{ij}(y)f_{kj}  (y+\alpha_{ij}-\alpha_{kj})\hspace{-2mm}\prod_{\ell\neq k,i}\hspace{-2mm}F_{\ell j}(y+\alpha_{ij}-\alpha_{\ell j})dy.\label{eq:term1_gradient}\numberthis
                  \end{align*}
                  Now we can express $\big|\frac{\partial q_{ij}}{\partial \alpha_{ij}}\big|$ and $\big|\frac{\partial q_{ij}}{\partial \alpha_{kj}}\big|\ \forall \ i,k\in[n]\ $ and $k\neq i$ as follows
                  \begin{align*}
                    &\hspace{27mm}\bigg|\frac{\partial q_{ij}}{\partial \alpha_{ij}}\bigg|=\bigg|\sum_{k\in [n]\backslash \{i\}}t(k)\bigg| \leq (n-1)\cdot\big|t(i)\big|\\
                    &\hspace{27mm}\hspace{10.5mm}
                        \stackrel{}{\leq}(n-1)\mu_{\max}\cdot\bigg|
                          \hspace{-4mm}\int\limits_{\mathrm{supp}(f_{ij})}
                          \hspace{-4mm}f_{ij}(y)dy\ \bigg|
                        &\tag{
                          ${\rm Using \ } f_{ij}(\phi_{ij})  \leq \mu_{\max} {\rm \ and \ }F_{ij}(\phi_{ij}) \leq 1$
                          }\\
                    &\hspace{27mm}\hspace{12.5mm}\hspace{-3mm}\stackrel{\eqref{eq:LemA8Eq1}}{\leq}(n-1)\mu_{\max}.
                        \tag{${\rm Using \ } \displaystyle\int\nolimits_{\mathrm{supp}(f_{ij})}\hspace{-4mm} f_{ij}(z)dz=1$, 87}
                        \addtocounter{equation}{1}
                        \label{eq:LemA10Bounda}\\
                    &\hspace{27mm}\hspace{-1mm}\bigg|\frac{\partial q_{ij}}{\partial \alpha_{kj}}\bigg|=\big|t(k)\big|\\
                    &\hspace{27mm}\hspace{-1.5mm}\hspace{10.5mm}\stackrel{\eqref{eq:cdf_is_bounded2}}{\leq}
                    \mu_{\max}\bigg|\int\limits_{\mathrm{supp}(f_{ij})}\hspace{-4mm}f_{ij}(y)dy\ \bigg|
                    \tag{${\rm Using \ } f_{ij}(\phi_{ij})  \leq \mu_{\max} {\rm \ and \ }F_{ij}(\phi_{ij}) \leq 1$}\\
                    &\hspace{27mm}\hspace{-1.5mm}\hspace{10.5mm}\stackrel{\eqref{eq:LemA8Eq1}}{\leq} \mu_{\max}.
                    \addtocounter{equation}{1}
                    \tag{${\rm Using \ } \displaystyle\int\nolimits_{\mathrm{supp}(f_{ij})}\hspace{-4mm} f_{ij}(z)dz=1$, 88} \label{eq:LemA10Boundb}
                  \end{align*}
                  Now we can show that the gradient of $q_{ij}(\alpha_{j})$ is bounded, i.e., $q_{ij}(\alpha_{j})$ is Lipschitz continuous. For this consider $\|\nabla q_{ij}(\alpha_{j})\|$
                  \begin{align*}
                    \|\nabla q_{ij}(\alpha_{j})\|_2^2 &=\hspace{4.5mm} \sum_{k\in[n]}\bigg(\ \bigg|\frac{\partial q_{ij}}{\partial \alpha_{kj}}\bigg|^2\ \bigg)\\
                    &\leq\hspace{4.5mm} \bigg|\frac{\partial q_{ij}}{\partial \alpha_{ij}}\bigg|^2+\sum_{k\in[n]\backslash \{i\}}\bigg(\ \bigg|\frac{\partial q_{ij}}{\partial \alpha_{kj}}\bigg|^2\ \bigg)\\
                    &\hspace{-4mm}\stackrel{(\hyperref[eq:LemA10Bounda]{87}),(\hyperref[eq:LemA10Boundb]{88})}{\leq} (n-1)^2\mu_{\max}^2+n\mu_{\max}^2\\
                    &\stackrel{}{\leq}\hspace{4.5mm} n^2\mu_{\max}^2.\numberthis\label{eq:LemA10Boundc}
                  \end{align*}
                  Since $q_{ij}(\alpha_{j})$ and $\delta_{ij}$ represent the probabilities of advertisers winning they sum to 1.
                  Therefore, for all user type $j\in [m]$, their sum is bounded by 1, i.e., $\sum_{i\in[n]}q_{ij}(\alpha_{j})\leq 1$ and $\sum_{i\in[n]}\delta_{ij}\leq 1$.
                  Using the triangle inequality we get
                  \begin{align*}
                    \sum_{i=1}^{n -1}|\delta_{ij}-q_{ij}(\alpha_{j})|&\leq \sum_{i=1}^{n -1}\bigg(
                    |\delta_{ij}|+|q_{ij}(\alpha_{j})|\bigg)\stackrel{}{\leq} 2.\numberthis\label{eq:LemmaA10Bound3}
                  \end{align*}
                  We represent the Hessian $\nabla^2 q_{ij}(\alpha_j)$ by $H(\alpha_j)$ for brevity. Then the Hessian of $\cL(\cdot)$ is
                  \begin{align*}
                    \nabla^2 \cL_j(\alpha_j) &= 2\cdot\sum_{i=1}^{n -1}\nabla q_{ij}(\alpha_j) \nabla q_{ij}(\alpha_j)^\top -\sum_{i=1}^{n -1} (\delta_{ij}-q_{ij}(\alpha_j))\cdot H(\alpha_j).\numberthis
                  \end{align*}
                  \noindent We know from Lemma~\ref{lem:smallLips} that the eigen-values of $H(\alpha_j)$ are bounded in absolute value by $2n(L+n\mu_{\max}^2)$.
                  We also know that the only non-zero eigen-value of $vv^\top$ for any vector $v$ is $\|v\|_2^2$.

                  Let $\|X\|_{\star}$ be the spectral-norm of matrix $X$, which is defined as the maximum singular value of $X$.
                  Then, since singular-values are absolute values of the eigen-values the spectral norm of $H(\alpha_j)$ and $vv^\top$ are bounded.
                  Specifically,
                  \begin{align*}
                    &\hspace{-15mm}\hspace{14mm}\|H(\alpha_j)\|_{\star}\stackrel{{\rm Lemma}~\ref{lem:smallLips}}{\leq} 2n(L+n\mu_{\max}^2)\numberthis\label{eq:spectral_bound_2}\\
                    &\hspace{-15mm}\|q_{ij}(\alpha_j)q_{ij}(\alpha_j)^\top\|_{\star}\hspace{6mm}\leq \|q_{ij}(\alpha_j)\|_2^2
                    \stackrel{\eqref{eq:LemA10Boundc}}{\leq} n^2\mu_{\max}^2.\label{eq:spectral_bound_1}\numberthis
                  \end{align*}
                  \noindent Now, we use the sub-additivity of the spectral-norm (represented as $\|\cdot\|_{\star}$)
                  \begin{align}
                    \hspace{38mm}\|A+B\|_{\star}\leq \|A\|_{\star}+\|B\|_{\star}. \labelthis{Sub-additivity of $\|\cdot\|_\star$}\label{eq:triangle_inequality_spectral_norm}
                  \end{align}
                  This gives us the following
                  \begin{align*}
                    \|\nabla^2 \cL_j(\alpha_j)\|_{\star}&\stackrel{(\hyperref[eq:triangle_inequality_spectral_norm]{94})}{\leq} 2\sum_{i\in [n-1]}  \|\nabla q_{ij}(\alpha_j) \nabla q_{ij}(\alpha_j)^\top\|_{\star}+ (\delta_{ij}-q_{ij}(\alpha_j)) \|H(\alpha_j)\|_{\star}\\
                    &\hspace{-3mm}\stackrel{\eqref{eq:spectral_bound_2},\eqref{eq:spectral_bound_1}}{\leq} 2\sum_{i\in [n-1]}  n^2\mu_{\max}^2+4n(L+n\mu_{\max}^2)\sum_{i\in[n-1]}(\delta_{ij}-q_{ij}(\alpha_j))\\
                    &\stackrel{\eqref{eq:LemmaA10Bound3}}{\leq} 2n^3\mu_{\max}^2+4n(L+n\mu_{\max}^2).
                  \end{align*}
                  Therefore, $\|\nabla^2 \cL_j(\alpha_j)\|_{\star}\leq O(n(L+n^2\mu_{\max}^2))$, and the eigen-values of $\|\nabla^2 \cL_j(\alpha_j)\|_{\star}$ are bounded in absolute value by $O(n(L+n^2\mu_{\max}^2))$.
                \end{proof}

                \section{Limitations and Future Work}
                  This work leaves several interesting directions open.
                  On the technical side, it would be interesting to improve Theorem~\ref{thm:2} by weakening the assumptions on the distributions, or by deriving better complexity bounds in terms of $\epsilon$ or $n$.
                  Although our algorithm works for intersectional types, it considers a separate constraint for each intersection.
                  $\text{Since there can be}$ exponentially many intersections compared to the types, it would be important to improve the run-time in this setting.
                  Exploring the utility lost from the advertiser's perspective, and potential ways of bounding it would also be of interest.
                  Further, it would be relevant to extend our framework to the (non-truthful) general second price auction~\cite{edelman2007internet, varian2009online}, which is used to auction multiple ad slots together.
                  From a practical standpoint, a natural problem is that advertisers run their campaigns at different times;
                  while an ad campaign is running on the platform, several other campaigns start and finish.
                  Our framework does not account for this.
                  Further, we do not ensure that users of different types derive similar value from an ad.
                  An advertiser could intentionally design an ad to appeal to a specific type, and then, even though the ad receives a balanced coverage, it could generate biased value for users~\cite{SpeicherAVRABGL18}.

                  Finally on the empirical side, testing our framework in the field and studying how the constraints affect user satisfaction, and the profile of ads they see would be important.

                  \section{Conclusion}
                    We initiate a formal study of designing mechanisms for online ad auctions that can ensure advertisements are not shown disproportionately to different populations.
                    This is especially relevant for ads for employment opportunities, housing, and other regulated markets where biases in the recipient population can be illegal and/or unethical.
                    As has been shown recently, existing platforms suffer from various spillover effects that result in such biased distributions.
                    Our approach places constraints on the allocations achieved by an ad across different sub-populations in order to ensure balanced exposure of the content.
                    It can be used flexibly placing constraints on some or all advertisers, across some or all sub-populations, and varying the tightness of the constraint depending on the level of fairness desired.

                    We present a truthful mechanism which attains the optimal revenue while satisfying the constraints necessary to attain such fairness, and present an efficient algorithm for finding this mechanism given the advertiser properties and fairness constraints.
                    Empirically, we observe that our mechanisms can satisfy fairness constraints at a minor loss to the revenue of the platform, even when the constraints ensure it attains perfect fairness.
                    Hence, fairness is not necessarily at odds with maximizing the platform's ad revenue.
                    Furthermore, we show empirically that advertisers are not significantly impacted with respect to their winning percentages -- the sub-populations their ads are shown to change to be fair, but overall they are still reach a similar number of users.


\nocite{scipy,scikit-learn}
\bibliography{bibliography}
\bibliographystyle{plain}

\appendix

\section{A Simple Example of Competitive Spillover}
  \begin{figure}[b!]
    \subfigure[
      \footnotesize
      \label{fig:8a} \textit{Coverage as a Function of Shift. (Non-Convex) }
      Coverage for one of the two advertisers with exponentially distributed bids, on two user types.
      We vary the shift of one of the advertisers and report its coverage as a function of the shift.
        ]
        {
          \includegraphics[height=5.0cm]{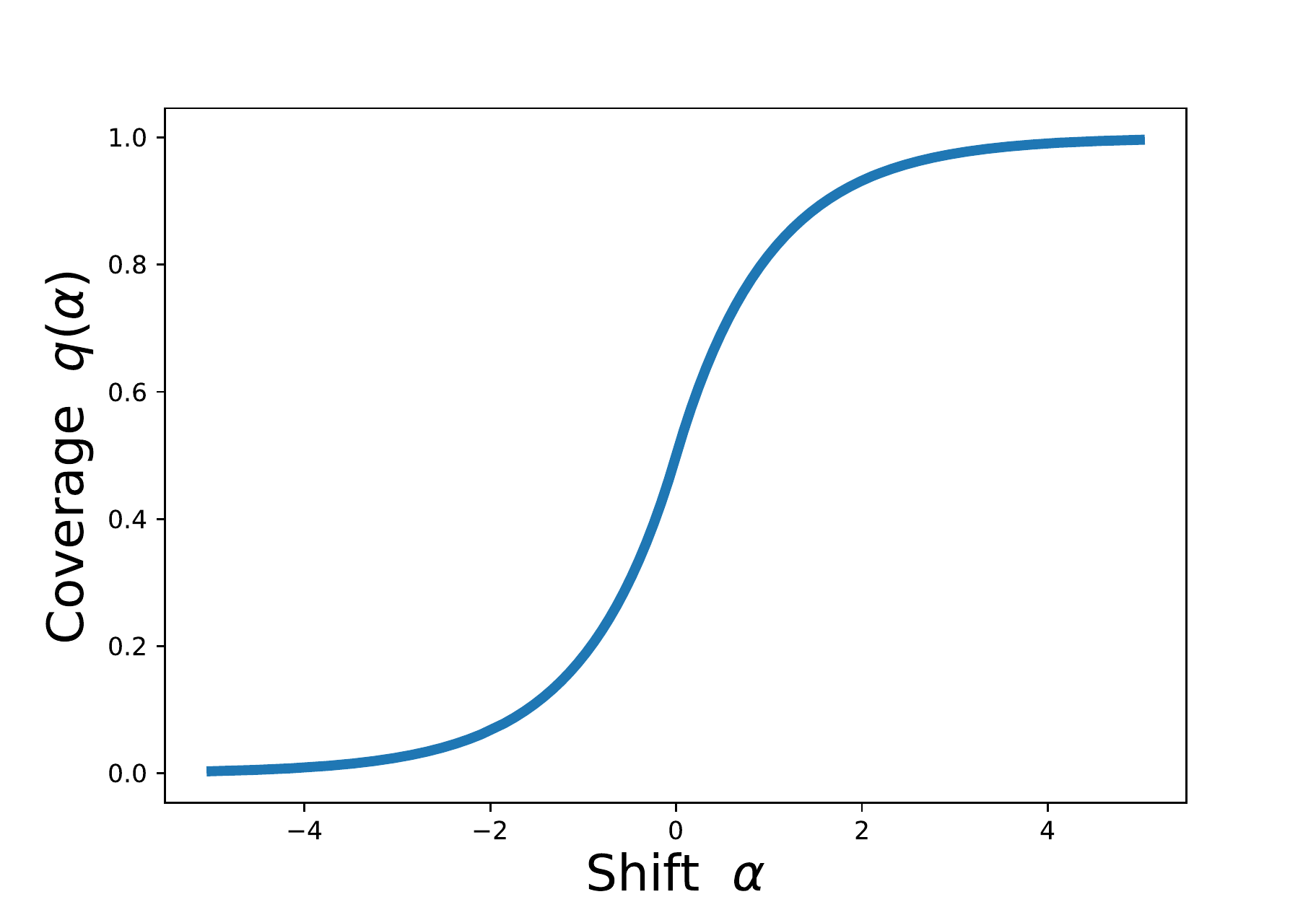}
        }
        \hspace{2mm}
    \subfigure[
      \footnotesize
        \label{fig:8b} \textit{Loss as a Function of Shifts. (Non-Convex) }
        The loss $\mathcal{L}(\alpha)$, for two advertisers with exponential valuations, and $\delta=(0.5,0.5)$.
        We vary the shift of one of the advertisers and report its coverage as a function of the shift.
        ]
        {
          \includegraphics[height=5.0cm]{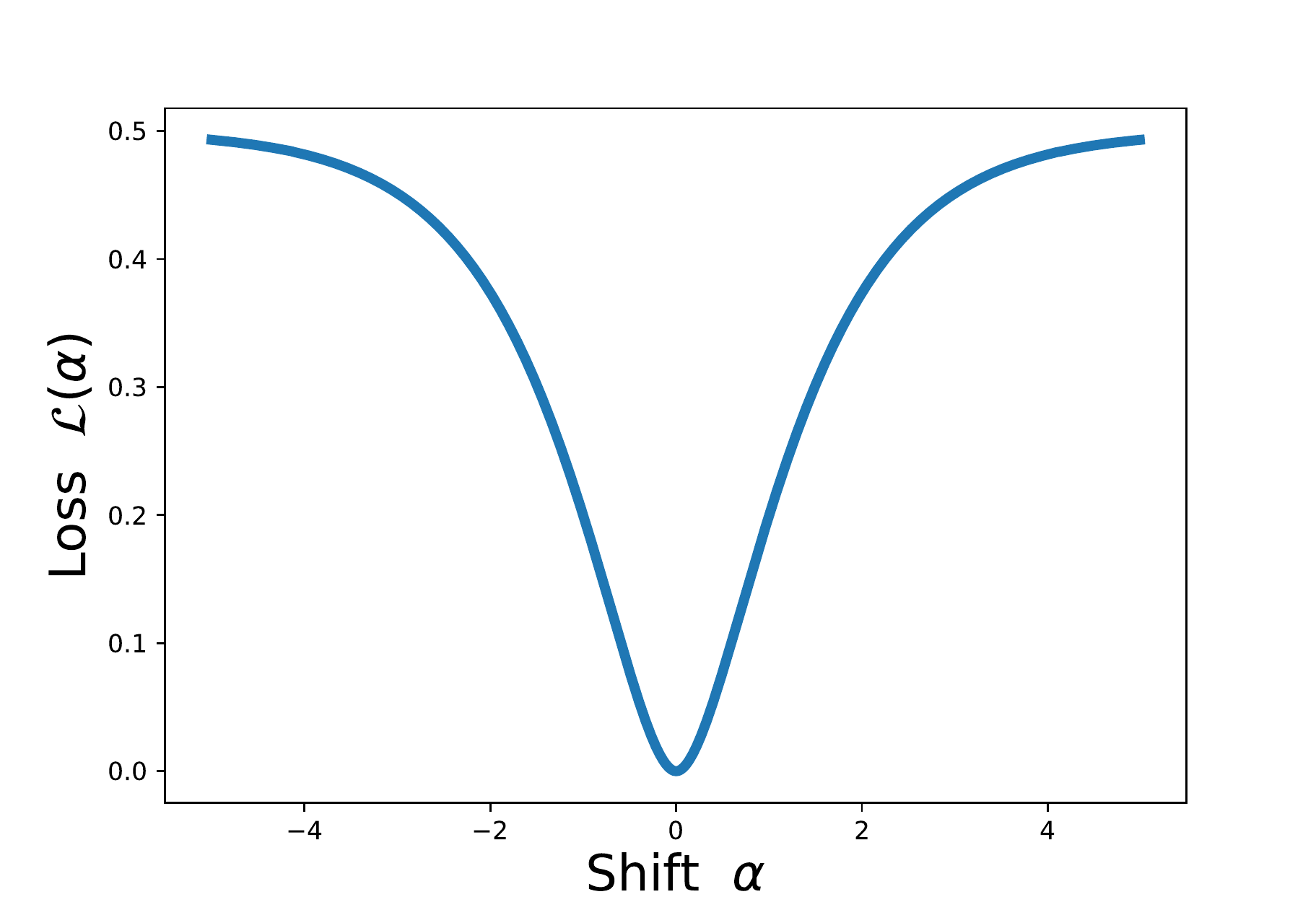}
        }
    \caption{
    \label{fig:8}
    }
  \end{figure}
  In this section we discuss the example of competitive spillover presented the introduction more concretely.

  Let the two advertisers have an equal budget of \$30.
  Both of them place a bid of \$1, if they target the current user and otherwise place a bid of \$0.
  For the sake of simplicity, let us assume that men and women visit the platform alternately.
  Consider an auction mechanism that shows the ad of the highest bidder.
  If there is more than one bidder with the same highest bid, the auction mechanism chooses one uniformly at random.
  Whenever a man visits the webpage, the second advertiser places a bid of \$1, while the first advertiser places a bid of \$0.
  Therefore, the mechanism always shows the second advertiser's ad.
  Whereas when a women visits the platform, both the advertisers place a bid of \$1, and one of them shows the advertisement with a 50\% probability.

  Consider the point when 40 users have visited the platform, 20 men and 20 women.
  The second advertiser has shown 20 ads to men, and 10 ads to women.
  Whereas the first advertiser has only displays 10 ads to females.
  Having shown 30 ads, the second advertiser has finished the budget, and leaves the auction, while the first advertiser stays till another 20 women visit the platform.
  In such a situation, the second advertiser who meant the ad to be unbiased among users, ends up under-representing women in the viewers of the ad.

\section{Revenue is Non-Concave in $\alpha$}
  \begin{figure*}[b!]
    \hspace{-6mm}
    \subfigure[
      \footnotesize
      \label{fig:9a}
        We report the total revenue as a function of the shift.
        ]
        {
          \includegraphics[height=5.5cm]{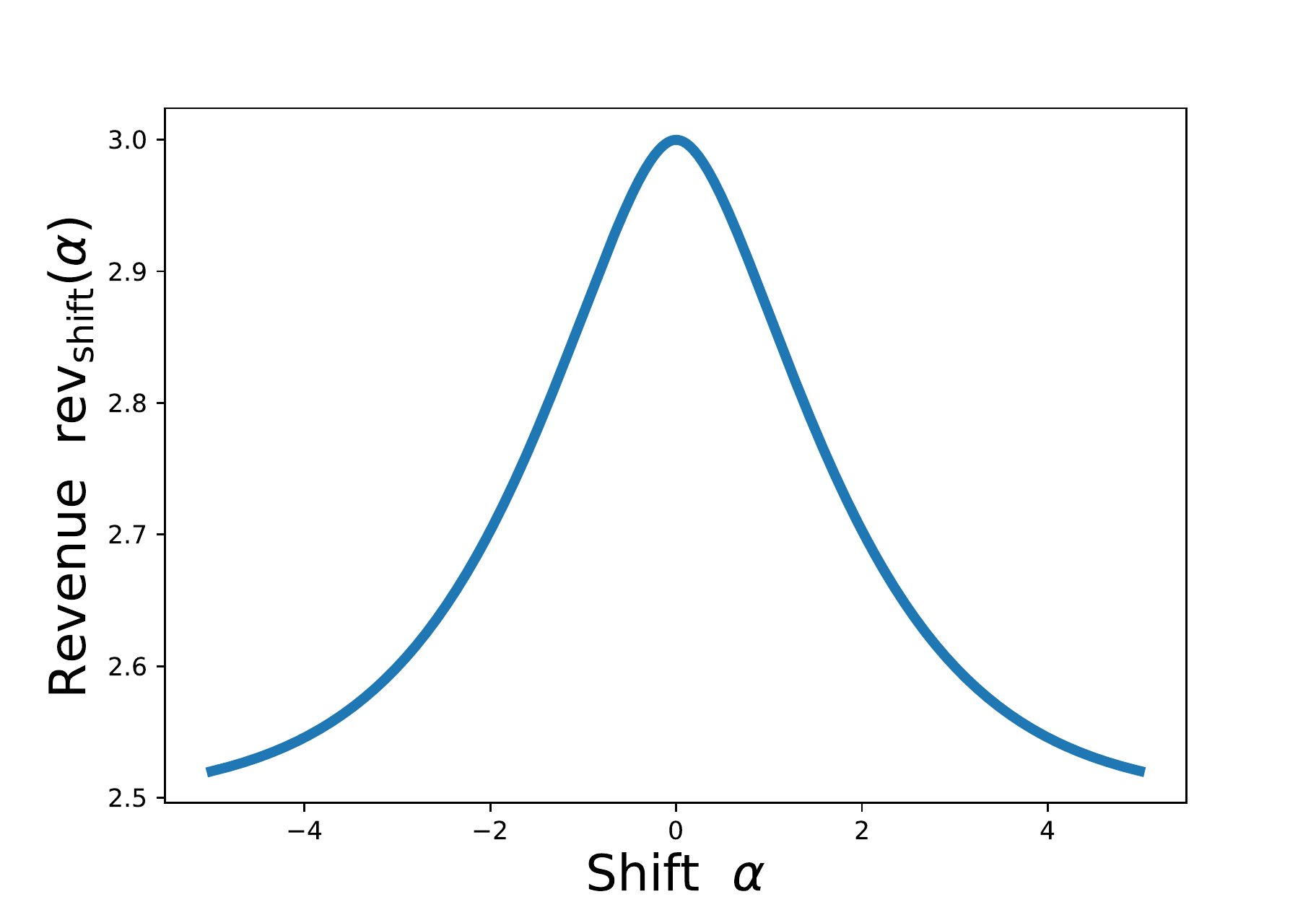}
        }
        \subfigure[
          \footnotesize
          \label{fig:9b}
            We report the total revenue as a function of the coverage.
            ]
            {
              \includegraphics[height=5.5cm]{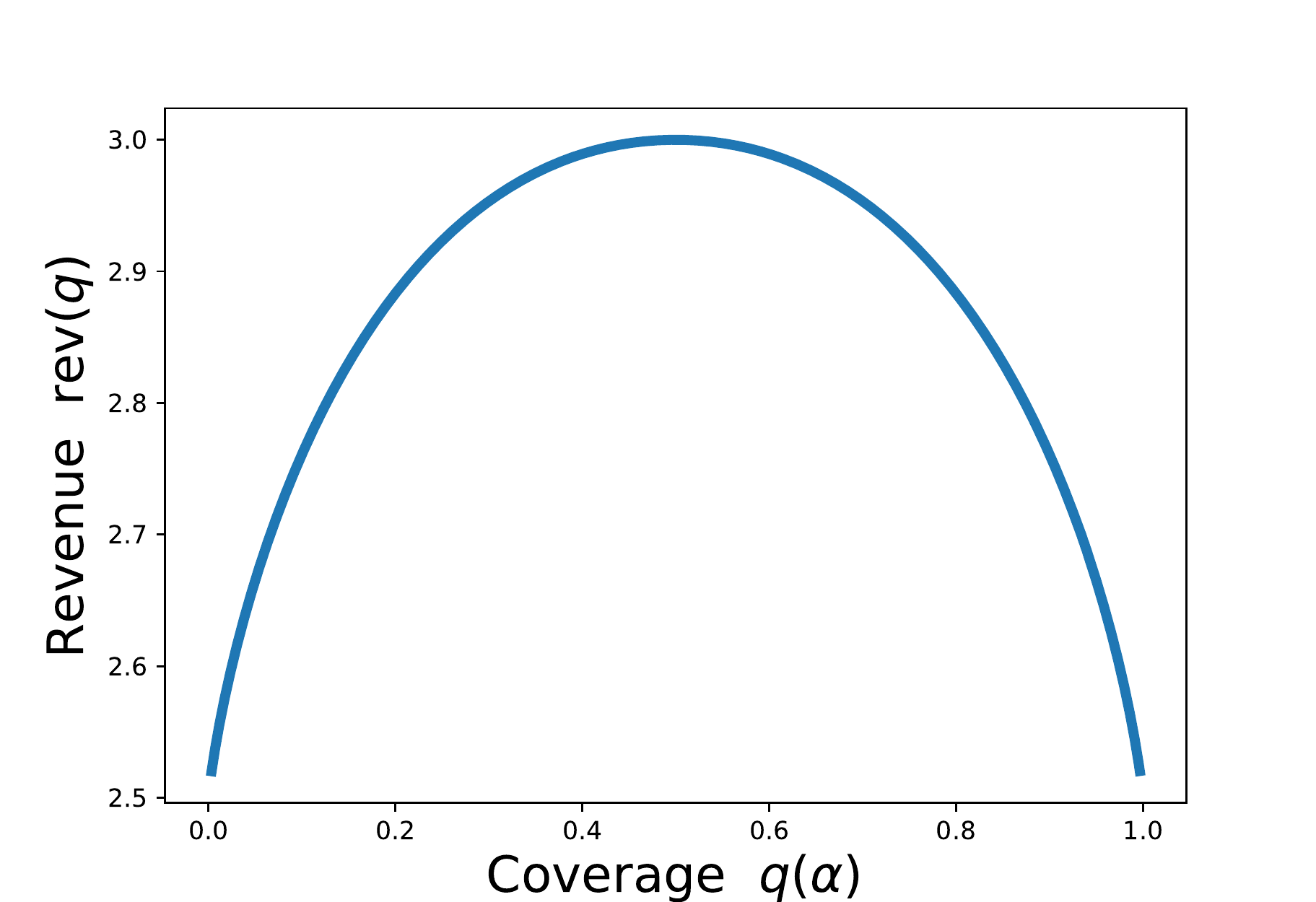}
              \hspace{2mm}
            }
    \caption{\label{fig:9} \textit{Revenue as a Function of Coverage and Shift.}
    Total revenue for two advertisers with exponentially distributed bids, on two user types. We vary the shift of one of the advertisers.}
  \end{figure*}
  \label{app:revenue_is_non_concave}
  Consider two advertisers and one user type with $f_{11}(x)=e^{-x}$ and $f_{21}(x)=e^{-x}$.
  We fix the shift of advertiser 2 to 0, and consider a positive shift $\alpha\geq 0$ of advertiser 1.
  Then
  \begin{align*}
    \mathrm{rev}_{\rm shift}(\alpha) &= \int\limits_{\mathrm{supp}(f_{11})}\hspace{-5mm}yf_{11}(y)F_{21}(y+\alpha)dy\hspace{0.5mm}+\hspace{-5mm}\int\limits_{\mathrm{supp}(f_{21})}\hspace{-5mm}yf_{21}(y)F_{11}(y-\alpha)dy\\
    &=\hspace{4.5mm}\int_0^{\infty}\hspace{-3mm}ye^{-y}(1-e^{-(y+\alpha)})dy+\hspace{-1mm}\int_{\alpha}^{\infty}\hspace{-4mm}ye^{-y}(1-e^{-(y-\alpha)})dy\\
    &=\hspace{4.5mm}1+\nfrac{1}{2}\cdot (\alpha+1) e^{-\alpha}.
  \end{align*}
  Differentiating $\mathrm{rev}_{\rm shift}$ we can observe it is not a concave function of the shift $\alpha$ (see Figure~\ref{fig:9a}).
  Indeed if we consider $\frac{d^2{\rm rev}_{\rm shift}}{d\alpha^2}=\nfrac{1}{2}\cdot (\alpha+1) e^{-\alpha}$, it is positive for all $\alpha>1$.
  Consider the coverage $q(\alpha)$ of advertiser 1
  \begin{align*}
    q(\alpha)&=\int\limits_{\mathrm{supp}(f_{11})}\hspace{-5mm}yf_{11}(y)F_{21}(y+\alpha)dy\\
    &=\hspace{4.5mm}\int_0^{\infty}\hspace{-4mm}e^{-y}(1-e^{-(y+\alpha)})dy\\
    &=\hspace{4.5mm}1-\nfrac{1}{2} \cdot e^{-\alpha}.
  \end{align*}
  Similarly we can observe that $q$ is not a convex function of $\alpha$ (see Figure~\ref{fig:8a}).
  Using $q(\alpha)$ to formulate the loss $\cL(\alpha)$ we can easily observe that it is non-convex as well (see Figure~\ref{fig:8b}).
  Let us re-parameterize the revenue ${\rm rev}_{\rm shift}$ in terms of $q$ as ${\rm rev}(\cdot)$. Then we have
  \begin{align*}
    &\hspace{46mm}\mathrm{rev}(1-q) = 1+(1-q)(1-\log(2-2q)))\numberthis\\
    &\hspace{46mm}\hspace{2.0mm}\frac{d^2\mathrm{rev}(q)}{dq^2}  = \frac{-1}{1-q} \leq 0. \hspace{57mm} (\mathrm{Using \ } q < 1)
  \end{align*}
  We can observe that revenue is a concave function of the coverage (see Figure~\ref{fig:9b}).

\section{Why Is the TV-Distance Small?}\label{sec:small_tv}
  To calculate the TV-distance we consider the distribution of winners selected by the auction mechanism, i.e., the distribution of the number of users an advertiser reaches.
  This distribution is different from coverage which separates the audience by their types.
  We report the total variation distance
  \begin{align}
    d_{TV}(\cM, \mathcal{F})\coloneqq \nfrac{1}{2}\sum_{i=1}^{n}|\sum_{j=1}^{m}q_{ij}(\mathcal{M})-q_{ij}(\mathcal{F})| \in [0,1]\label{eq:tv_norm_supp}
  \end{align}
  between the two distributions, as a measure of how much the winning distribution changes due to the fairness constraints.

  Consider the distribution of advertiser $i$'s coverage as the vector $\{q_{ij}(\mathcal{M})\}_{j\in[m]}\in [0,1]^m$.
  Its projection on the perfectly-fair polytope is
  \begin{align}
    \frac{1}{m} \bigg( \ \sum_{j=1}^{m}q_{ij}(\mathcal{M}) \ \bigg)\cdot (1,1,\cdots,1)\in [0,1]^{n}.\label{eq:projection_of_optimal_distribution}
  \end{align}
  Since the coverage is uniform, it satisfies the perfect fairness constraints.
  Further, using Eq.~\ref{eq:tv_norm_supp} we can observe that this projection has a 0 total variation distance $d_{TV}$ to $\{q_{ij}(\cM)\}_{j\in[m]}$.

  If the solution $q_{ij}(\mathcal{F})$ of the optimal fair mechanism is close to this projection, then the resulting $d_{TV}(\cM, \mathcal{F})$ is small.
  Moving the coverage $q_{ij}(\mathcal{F})$ away from the projection involves a trade-off between increasing the total change in coverage, and decreasing the change for some types the advertiser values more.

  Therefore, if the average bid of an advertiser does not vary significantly between the types, then $q_{ij}(\mathcal{F})$ is close to the projection.
  Importantly, this does not imply that the coverages $q_{ij}(\cM)$ of the unconstrained mechanism are balanced.
  To gain some intuition, consider two advertisers with similar budgets, but one advertiser places a bid of $1+\epsilon$ for men and $1-\epsilon$ for women, while the other places a bid of $1$ for men and women.
  Even though the first advertiser's bid for men is only $2\epsilon$ higher than their bid for women, they would only reach men, i.e., $q_{1}=(1,0)$.
  Whereas, the platform only loses just $\epsilon$ fraction of its revenue by changing $q_{1}$ to its projection $(\nfrac{1}{2},\nfrac{1}{2})$.

\end{document}